\documentclass{article}

\usepackage[utf8]{inputenc}
\usepackage{amsmath, amssymb}
\usepackage{graphicx}
\usepackage{hyperref}
\usepackage{xcolor}
\usepackage{amsthm}

\usepackage[
    backend=biber,style=authoryear, maxcitenames=1, uniquename=false, uniquelist=false
]{biblatex}
\newtheorem{definition}{Definition}
\newtheorem{proposition}{Proposition}
\newtheorem{lemma}{Lemma}
\newtheorem{theorem}{Theorem}
\newtheorem{corollary}{Corollary}
\usepackage{algorithm}
\usepackage[noend]{algpseudocode}
\usepackage{caption}
\usepackage{subcaption}

\usepackage{geometry}
\usepackage[utf8]{inputenc} 
\usepackage[T1]{fontenc}    
\usepackage{hyperref}       
\usepackage{url}            
\usepackage{booktabs}       
\usepackage{amsfonts}       
\usepackage{nicefrac}       
\usepackage{microtype}      
\usepackage{lipsum}
\usepackage{graphicx}
\graphicspath{ {./images/} }
\usepackage{lineno}
\usepackage{hyperref}

\usepackage{verbatim}

\def\PP{\mathbb{P}} 

\newtheorem{assumption}{}

\newcommand{\asec}{\texttt{ASE-CLUST}}

\newcommand{\dasec}{\texttt{DASE-CLUST}}

\title{A Doubled Adjacency Spectral Embedding Approach to Graph Clustering}

\author{
  Sinyoung Park \\
  Department of Mathematical Sciences
  \\
  University of Bath \\
  \texttt{sp2681@bath.ac.uk}
  \and
  Matthew A. Nunes \\
  Department of Mathematical Sciences \\
  University of Bath \\
  \texttt{man54@bath.ac.uk}
  \and
  Sandipan Roy \\
  Department of Mathematical Sciences \\
  University of Bath \\
  \texttt{sr2081@bath.ac.uk}
}

\begin{document}
\maketitle

\begin{abstract}
Spectral clustering is a popular tool in network data analysis, with applications in a variety of scientific application areas. However, many studies have shown that classical spectral clustering does not perform well on certain network structures, particularly core-periphery networks. To improve clustering performance in core-periphery structures, Adjacency Spectral Embedding (ASE) has been introduced, which performs clustering via a network's adjacency matrix instead of the graph Laplacian. Despite its advantages in this setting, the optimal performance of ASE is limited to dense networks, whilst network data observed in practice is often sparse in nature. To address this limitation, we propose a new approach which we term Doubled Adjacency Spectral Embedding (DASE), motivated by the observation that the squared adjacency matrix will leverage the fewer connections in sparse structures more efficiently in clustering applications. Theoretical results establish that the resulting clustering algorithm enjoys good consistency properties when determining sparse community structure. The performance and general applicability of the proposed method is evaluated using extensive simulations on both directed and undirected networks. Our results highlight the improved clustering performance on both sparse and dense networks in the presence of core-periphery structures. We illustrate our proposed technique on real-world employment and transportation datasets. 
\end{abstract}

\providecommand{\keywords}[1]{\noindent\textbf{Keywords:} #1}
\keywords{Adjacency Spectral Embedding; clustering; community detection; core-periphery networks.}

\section{Introduction}\label{sec:intro}

Network data are ubiquitous in numerous disciplines, including politics \parencite{adamic2005political}, neuroscience \parencite{rubinov2010complex}, biology \parencite{pavlopoulos2011using}, power systems \parencite{pagani2013power}, ecology \parencite{wey2008social}, and transportation \parencite{derudder2007flying}. These structures provide a natural framework for analysing, interpreting, and predicting complex relationships within these domains. For example, in social network analysis, individuals are typically represented as nodes, while their interactions are modelled as directed or undirected edges. Such graph-based representations facilitate identification of patterns and emergent behaviours, thereby enabling rigorous analysis and deeper interpretation of underlying dependence relationships. A critical aspect of network analysis lies in understanding their intrinsic structure. Networks can exhibit various structural patterns, including assortative, non-assortative, scale-free, and core–periphery configurations. 

One of the most fundamental tasks in analysing networks is to find clusters or communities that group similar objects/individuals together. The mathematical approaches for identifying such groups are commonly known as clustering or community detection techniques. Identifying communities within a network enables deeper understanding of local structures (dyads, triads etc.) as well as a holistic (global) understanding of the entire network. 
Community detection has been extensively studied in the statistical learning literature, see e.g. \cite{fjallstrom1998algorithms, newman2004finding, fortunato2010community, goldenberg2010survey, abbe2018community} and the references therein for overviews. Several approaches have been proposed to identify communities in networks, including spectral embedding \parencite{von2007tutorial, rohe2011spectral, rubin2022statistical}. Spectral clustering (SC) and its variants have developed over recent years as some of the most popular techniques for finding communities in network data, and their use is now commonplace in a number of settings, see e.g., \cite{shi2000normalized, ng2001spectral, von2007tutorial}. These methods predominantly perform spectral embedding using the Laplacian matrix -- a representation of the connections between individuals in network data.  Most of this work has been dedicated to undirected graphs, but there has been a growing body of literature on community detection in directed networks \parencite{satuluri2011symmetrizations, cucuringu2020hermitian, wang2020spectral}.  

Recent work has shown that Laplacian-based spectral embedding (LSE) tends to perform well in assortative or non-assortative networks, but not in other settings such as core-periphery structures \parencite{priebe2019two}; these networks naturally have two distinct regions: the `core' within which nodes are densely interconnected, and the `periphery' which exhibits sparse internal connections \parencite{borgatti2000models}.  Such structures appear in a number of application areas and data settings, such as international relations \parencite{chase2019core}, co-citation networks and social interaction networks \parencite{tang2019recent} as well as economics \parencite{in2020formation}, and have been studied more intensively in recent years, predominantly related to the identification of core structures (see e.g. \cite{csermely2013structure, rombach2014core, rombach2017core, zhang2015identification}).  To address the unsatisfactory performance of Laplacian-based spectral clustering in this setting, Adjacency Spectral Embedding (ASE) has been proposed \parencite{sussman2012consistent}. In particular, it has been shown that ASE provides a more meaningful interpretation of the underlying community structure than traditional spectral clustering for core-periphery networks, under a general stochastic block modelling framework \parencite{priebe2019two, tang2018limit}. However, ASE can struggle to detect structures when the underlying network is sparse.  For the specific case of {\em directed} networks, recent research has proposed statistical inference methods that can detect core-periphery structures in large networks \parencite{priebe2019two, zhang2015identification}, and directed core-periphery structures with L-shapes \parencite{elliott2020core}.

In this article, we develop an alternative approach to clustering nodes in a graph which is suitable for both undirected and directed core-periphery network structures, based on spectral embedding of the {\em squared} adjacency matrix, which we term \textit{Double Adjacency Spectral Embedding} (DASE). Our approach has similar motivation to the work of \cite{rohe2011spectral}, who show that spectral clustering based on the squared Laplacian achieves good consistency and misclustering properties under the stochastic block model in assortative / non-assortative settings. However, our work differs in that our focus is on core-periphery networks, for which adjacency-based embedding is more appropriate. Combined with a clustering algorithm such as $k$-means, our proposed clustering algorithm \dasec\ provides consistent clustering of nodes in a network under the general stochastic block model, and is particularly suitable for core-periphery structures, whether directed or undirected. Our proposed approach shows improved recovery of the true community structure compared to the ASE technique of \cite{sussman2012consistent}, particularly for challenging setups. One can show that DASE can be thought as a variant of an embedding under the weighted stochastic block model \parencite{gallagher2024spectral}, and that DASE provides more informative embeddings than ASE in certain settings as measured by the Chernoff information \parencite{chernoff1952measure}, commonly used to assess embeddings \parencite{tang2018limit, gallagher2024spectral}.

This article is structured as follows. In Section \ref{sec:background}, we introduce the stochastic block model (SBM), upon which we build our proposed methodology, and define core-periphery networks which we focus on in this work. Our proposed Doubled Adjacency Spectral Embedding (DASE) approach is described in Section \ref{sec:methodology}, whilst Section \ref{sec:theory} establishes theoretical properties related to DASE as well as the subsequent clustering algorithm, both in the general and core-periphery network cases. We demonstrate the efficacy of the DASE embedding and our proposed clustering procedure through extensive simulation studies in Section \ref{sec:sims}, considering both directed and undirected networks. Finally, in Section \ref{sec:data}, we consider the application of our proposed clustering methodology to two real-world datasets, namely faculty hiring data and flights observed on an airport network. In Section \ref{sec:concs}, we make some concluding remarks.


\section{Background}\label{sec:background}

In this section, we provide an overview of the stochastic block model \parencite{holland1983stochastic, Wang01031987}, and the random dot product graph embedding framework as developed by \cite{young2007random}, which we consider in this work. We also provide a summary of the adjacency spectral clustering algorithm of \cite{sussman2012consistent}, against which we benchmark our proposed techniques in Section \ref{sec:ase}. These aspects provide the necessary background to formulate our proposed adjacency-based embedding and clustering algorithm, described in Section \ref{sec:methodology}; we also recall the notion of core-periphery networks, for which structure our proposed techniques are particularly well-suited. \\ 


 For the remainder of this article, we denote a graph (or network) with $N$ nodes by $\mathcal{G}=(\mathcal{V},\mathcal{E})$, where $\mathcal{V}=\{1,2,\ldots,N\} = [N]$ is the vertex set and $\mathcal{E}=\{(i,j): i\in \mathcal{V}, j\in \mathcal{V}, i\neq j\}$ is its edge set, with the number of edges is given by the cardinality of $\mathcal{E}$, denoted $|\mathcal{E}|$. The edge set of $\mathcal{G}$ is commonly represented in the adjacency matrix $A \in \{0, 1\}^{N\times N}$, where the matrix entry $A_{ij}$ indicates whether or not there is an edge from node $i$ to node $j$. Note that for undirected graphs, $A$ is symmetric, i.e., $A_{ij}=A_{ji}$.  In this work, we consider graphs with no self-loops, i.e., $A_{ii} = 0$ for all $i \in [N]$.
 
\subsection{Stochastic Block Modelling Framework}\label{sec:sbm}

The stochastic block model (SBM) \parencite{holland1983stochastic} has been established as one of the most popular probabilistic generative models for networks. It has an inherent block structure, where block (community / cluster) membership is determined solely by a latent vector. The model is defined as follows.

\begin{definition}[\cite{holland1983stochastic}]
\label{def:sbm}
    For a graph $\mathcal{G}$ with nodes $i \in [N]$, let $\theta_i \in [K]$ denote the (unique) block membership label of node $i$, where $K$ is the number of blocks. $\{ \theta_i \}_{i=1}^N$ are independent with $\mathbb{P}(\theta_i = k) = \pi_k$. Given a block probability matrix $B \in [0,1]^{K \times K}$ and a parameter vector $\pi\in(0,1)^{K}$, the stochastic block model (SBM) generates the adjacency matrix $A$ conditional on $\theta$ as
    \begin{equation}\label{eq:sbm}
        A_{ij}|\theta_i,\theta_j \overset{\text{ind}}{\sim} \text{Bernoulli} \ (B_{\theta_i \theta_j}), \qquad i \ne j.
    \end{equation}

\end{definition}
\bigskip

In Definition \ref{def:sbm}, for undirected graphs, $A_{ij}$ is generated only for $i < j$ with the symmetric block matrix $B$; for directed graphs, all ordered pairs $(i,j)$ are generated independently. The vector $\pi \in (0,1)^K$ denotes the vector of group proportions, with the entries summing to $1$. For example, when $K=2$, the case $\pi = (0.5, 0.5)^\top$ represents a balanced network.


Note that we can write the likelihood of entries of the adjacency matrix under the model as
\begin{equation} \label{sbmprob}
\PP(A|\theta)  =\prod_{i\neq j}\PP(A_{ij}|\theta_i,\theta_j)
 = \prod_{i\neq j}(B_{\theta_i\theta_j})^{A_{ij}}(1-B_{\theta_i\theta_j})^{1-A_{ij}}.
\end{equation}

We define the node-wise probability matrix $Q \in [0,1]^{N \times N}$ as $Q =ZBZ^\top$ for both undirected and directed adjacency matrices, where $Z \in \{0,1\}^{N \times K}$ is the membership indicator matrix: $Z_{ij} = 1$ if the node $i$ belongs to the $j$-th group. Note that with the independence assumption in Definition \ref{def:sbm}, the entries of $A$ (the presence or absence of edges) are conditionally independent given $Q$.

\subsection{The Random Dot Product Graph Model}\label{sec:rdpg}

The stochastic block model can be seen as an example of the more general latent space network model of \cite{hoff2002latent}, see e.g. \cite{rohe2011spectral}.  A particularly popular variant of a latent space model is the random dot product graph (RDPG) model introduced by \cite{young2007random}. The RDPG model defines the propensity of a connection between two nodes in a network via the dot product of their 
latent vectors; the SBM can also be alternatively characterised using the RDPG model, facilitating a probabilistic interpretation of the graph adjacency via particular latent space dot product embeddings.  


\begin{definition}[\cite{young2007random}]
\label{def:rdpg}
    Let $X, Y \in \mathbb{R}^{N \times d}$, where $X_i \in \mathbb{R}^{d}$ and $Y_i \in \mathbb{R}^{d}$ represent latent position vectors for nodes $i \in [N]$, with latent space dimension $d$. Assume that the latent position matrices $X$ and $Y$ are random and satisfy
    \[
    \mathbb{P} \left( \langle X_i, Y_j \rangle \in [0,1] \right) = 1, \qquad \text{for all } i,j \in [N].
    \]
    Then, under the Random Dot Product Graph (RDPG) model, conditioned on $X_i$  and $Y_i$ the entries of the adjacency matrix $A_{ij}$ are independent Bernoulli random variables with parameter $\langle X_i, Y_j \rangle$ for all $i \ne j$, so that
    \begin{equation}
        \mathbb{P} \left( A|X, Y \right) = \prod_{i \ne j} \mathbb{P} \left( A_{ij} | X_i, Y_j \right).
    \end{equation}
    For the adjacency matrix $A$ of undirected graphs, we have 
    \begin{equation}
        \mathbb{P} \left( A|X, Y \right) = \prod_{i < j} \mathbb{P} \left( A_{ij} | X_i, X_j \right)
    \end{equation}
\end{definition}
\bigskip

Based on Definition \ref{def:rdpg}, the stochastic block model can be viewed as a special case of the RDPG model \parencite{athreya2018statistical, sussman2012consistent}. In particular, for directed graphs,
\begin{equation}
    \mathbb{P}(A_{ij} = 1) = B_{\theta_i \theta_j} = \langle X_i, Y_j \rangle,
\end{equation}
where the latent dimension $d$ equals the rank of the block probability matrix $B$. 

\subsection{Measuring Performance of Clustering}\label{sec:measuringperformance}

As discussed in detail in \cite{sussman2012consistent}, for a graph distributed according to the model defined in \ref{def:sbm}, 
we can use the following mean square error criterion for 
clustering the rows of $\hat{X}$ into $K$ blocks:
\begin{equation}
\label{eq:clust}
(\hat{\nu}, \hat{\theta}) 
= \arg\min_{\nu,\, \theta} 
\sum_{u=1}^{N} \left\| \hat{X}_{u} - \nu_{\theta(u)} \right\|_2^2, 
\end{equation}
where $\hat{\nu}_i \in \mathbb{R}^m$ gives the centroid of block $i$,
$\theta : [N] \to [K]$ is the block assignment function and $\hat{X} \in \mathbb{R}^{N \times m}$ is the latent embedding matrix with embedding dimension $m \in \mathbb{N}$. Note that this criterion corresponds to that used by $k$-means clustering as used below in Algorithms \ref{alg:ASE} and \ref{alg:DASE}. 

\subsection{Graph Clustering with Adjacency Spectral Embedding}\label{sec:ase}

Traditional spectral clustering is achieved by taking the singular value decomposition (SVD) of the Laplacian matrix, $L$, and then performing a clustering algorithm such as $k$-means \parencite{macqueen1967multivariate, lloyd1982least} using the singular vectors corresponding to the largest $d$ singular values of $L$, where $d$ is some chosen latent dimension. As an alternative embedding and clustering algorithm, \cite{sussman2012consistent, priebe2019two} propose taking the SVD of the adjacency matrix, $A$, instead of $L$, termed {\em Adjacency Spectral Embedding} (ASE). The clustering algorithm based on ASE is summarised in Algorithm \ref{alg:ASE}; in this article, we refer this clustering procedure as \asec\ to distinguish it from the embedding itself.

\begin{algorithm}[H]
\caption{Adjacency Spectral Clustering (\texttt{ASE-CLUST})}
\label{alg:ASE}
\hspace*{\algorithmicindent} 
\textbf{Input}: Adjacency matrix $A \in \{ 0, 1 \}^{N \times N}$, number of clusters $K \in \mathbb{N}$ and latent dimension $d \in \mathbb{N}$\\
 \hspace*{\algorithmicindent} 
 \textbf{Output}: Clustering labels $l \in \mathbb{R}^N$
\begin{algorithmic}[1]
\State Compute the singular value decomposition (SVD) of $A$, i.e., $A = U \Sigma V^T$, where the singular values in $\Sigma$ are ordered in decreasing magnitude.
\If{$A$ is directed}
    \State Define the embedding $S = \left[ U_d\Sigma_d^{1/2}| V_d\Sigma_d^{1/2} \right] \in \mathbb{R}^{N \times 2d}$, where $U_d$ and $V_d$ contain the top $d$ left and right singular vectors, respectively. $\Sigma_d$ contains the corresponding singular values.
\Else
    \State Define the embedding $S = U_d \Sigma_d^{1/2} \in \mathbb{R}^{N \times d}$, where $U_d$ contains the top $d$ eigenvectors corresponding to the largest $d$ singular values in $\Sigma_d$.
\EndIf
\State Extract the clustering labels $l$ from an algorithm such as $k$-means (with $K$ clusters) applied to the rows of $S$, using the criterion \eqref{eq:clust}. 
\end{algorithmic}
\end{algorithm}
\bigskip

Note that other unsupervised clustering algorithms can be substituted for $k$-means in Step 6 of Algorithm \ref{alg:ASE}. In particular, Gaussian mixture models (GMMs) \parencite{mclachlan1988mixture, dempster1977maximum} are a popular alternative (see e.g. \cite{priebe2019two}); results in \cite{tang2018limit} have shown that asymptotically, spectral clustering methods such as LSE and ASE can be expressed as GMMs.  The GMM formulation of spectral clustering methods can be used to compute the Chernoff information for assessing the quality of embeddings, see e.g. \cite{tang2018limit, priebe2019two, gallagher2024spectral}.  We explore this aspect of our proposed clustering methodology in Section \ref{sec:chernoff}. 

\subsection{Core-periphery Graphs}\label{sec:cpgraphs}

In this paper, we focus on core-periphery structures, for which adjacency-based embeddings have been shown to be successful. Several intuitive interpretations of core-periphery graphs have been summarised in \cite{borgatti2000models, rombach2014core}. A common characterisation of such networks in previous research is that the connectivity structure separates into two distinct groups, as follows.

\begin{definition}[Core-periphery graphs]
\label{def:core-periphery}
    Core-periphery graphs have connectivity structures consisting of two parts: (i) the `core', within which connections are dense, and (ii) the `periphery', within which, relative to the core structure, connections are sparse. In addition, nodes in the core group have relatively dense connections with nodes in the periphery.
\end{definition}
\bigskip

Definition \ref{def:core-periphery} has a natural interpretation for a two-cluster ($K=2$) block structure. Recent research has broadened the concept of core-periphery graphs for $K>2$. For instance, \cite{rombach2014core} introduced models with multiple core groups, whilst \cite{gallagher2021clarified} proposed a hub-and-spoke model that represents a layered structure, including an intermediate layer between core and periphery groups. These developments highlight a fundamental challenge in defining core-periphery graphs: for $K>2$, existing approaches either generalise the block structure or adopt a layered structure. For this reason, developing a single, consistent framework that encompasses both structural definitions for $K>2$ is challenging as these definitions are contrasting. We therefore restrict our focus in this paper to core-periphery graphs with $K=2$ as defined in Definition \ref{def:core-periphery}, leaving extensions to $K>2$ as future work.

For our setting, we consider an adjacency matrix $A$ generated from an SBM with two communities where the probability matrix $B$ is a $2 \times 2$ matrix. The within-core and within-periphery connection probabilities are $p$ and $s$, respectively, while the between group connection probabilities are $q$ and $r$. We assume the ordering $s<r,q<p$, which implies that Cluster $1$ has the highest within-cluster connectivity and thus represents the core group. Figure \ref{Fig:heatmap_core-periphery} illustrates this block probability structure.


\begin{figure}
    \centering
    \includegraphics[width=0.4\linewidth]{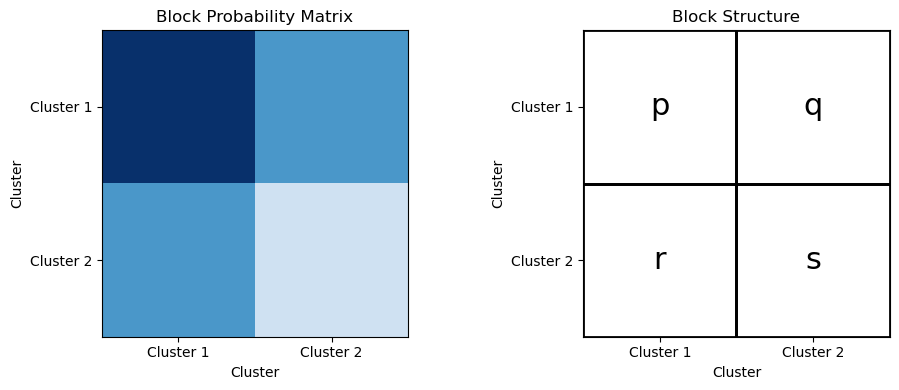}
    \caption{Block connection probability structure for core-periphery graphs.}
    \label{Fig:heatmap_core-periphery}
\end{figure}

\section{Graph Clustering with Doubled Adjacency Spectral Embedding}
\label{sec:methodology}

In this section, we introduce our new spectral embedding methodology and clustering technique based on a latent space representation of the squared adjacency matrix of a graph. More specifically, we consider the properties of the embedding in Section \ref{sec:dase}; the subsequent spectral clustering procedure, \dasec, is described in Section \ref{sec:daseclust}. 

\subsection{Doubled Adjacency Spectral Embedding (DASE)} \label{sec:dase}
In the following, we introduce our square-transformed adjacency matrix which we will use to perform our clustering to detect core and periphery communities.  We begin by discussing the doubled adjacency matrix.

\begin{definition}[Doubled adjacency matrix, $\tilde{A}$]
    Given an adjacency matrix $A \in \{0, 1\}^{N \times N}$, the doubled adjacency matrix $\tilde{A} \in [0,N]^{N \times N}$ is defined as
    \begin{equation}\label{eq:doubleA}
        \tilde{A} := AA,\ \textrm{with} \ [AA]_{ij} = \sum_{k=1}^N A_{ik} A_{kj}.
    \end{equation}
\end{definition}
\bigskip

The doubled adjacency matrix $\tilde{A}$, which is simply the square of the adjacency matrix, represents the number of two-step walks that one can take from one node to another for both directed and undirected graphs. For example, since $A$ indicates the connections between nodes, $(\tilde{A})_{12}$ indicates that the number of ways to travel from node $1$ to node $2$ in two steps.

In contrast to the adjacency matrix $A$, whose entries are Bernoulli distributed under the SBM (Equation \eqref{eq:sbm}), each entry of $\tilde{A}$ is the sum of products of Bernoulli random variables, and (under the assumption of independence such as in Definition \ref{def:sbm}) follows a Poisson-Binomial (PB) distribution \parencite{tang2023poisson}.   More specifically, if $A$ is the adjacency matrix of a directed graph, $A_{is}A_{sj}$ and $A_{it}A_{tj}$, are conditionally independent for $s \ne t$ given any $i$ and $j$, since the entries of $A$ are conditionally independent Bernoulli random variables with parameters $Q_{ij} = \mathbb{E}(A_{ij}) = (XY^\top)_{ij}$, where $X$ and $Y$ are the matrices of latent vectors as in the random dot product graph model formulation of the SBM in Section \ref{sec:rdpg}. Therefore,
\[
\tilde{A}_{ij} \sim \text{PB}(Q_{i1}Q_{1j}, \dots, Q_{iN}Q_{Nj}).
\]
If $A$ is the adjacency matrix of an undirected graph, then for the diagonal entries $(i=j)$, we have
\[
\mathbb{E}(\tilde{A}_{ii}) = \tilde{Q}_{ii} = \sum_{k=1}^N Q_{ik},
\]
and for off-diagonal entries $(i \ne j)$, 
\[
\mathbb{E}(\tilde{A}_{ij}) = \tilde{Q}_{ij} = \sum_{k=1}^N Q_{ik} Q_{kj}.
\]
Thus, for undirected graphs
\begin{align*}
    \tilde{A}_{ij} \sim
    \begin{cases}
        PB(Q_{i1}Q_{1j}, \dots, Q_{iN}Q_{Nj}) \quad & \text{if } i \ne j, \\
        PB(Q_{i1}, \dots, Q_{iN}) \quad & \text{if } i = j.
    \end{cases}
\end{align*}

Beyond the interpretation of $\tilde{A}$ as counting two-step walks on the graph, the doubled adjacency matrix can also be viewed as an instance of the weighted Generalized Random Dot Product Graph (WGRDPG) Model \parencite{gallagher2024spectral}, in which conditional on the block membership labels $\theta_i$, $\theta_j$, a general real-valued matrix $G$ (not necessarily the adjacency) is allowed to have any real-valued distribution, i.e.,  
    \begin{equation}\label{eq:wrdpg}
        G_{ij}|\theta_i,\theta_j \overset{\text{ind}}{\sim} \mathcal{H} (\theta_i,\theta_j), 
    \end{equation}
in place of the Bernoulli distribution in \eqref{eq:sbm} and the subsequent discussion in Section \ref{sec:rdpg}.  In our case, the matrix $G$ in \eqref{eq:wrdpg} is taken to be $\tilde{A}$, which reflects the strength of connections between nodes, and the weights follow a Poisson-Binomial distribution. 

Similar to Adjacency Spectral Embedding (ASE) which uses the latent space representation of the adjacency matrix $A$ under the RDPG model, we refer to the latent space representation of $\tilde{A}$ under the (weighted) random dot product graph model as {\em Doubled} Adjacency Spectral Embedding (DASE); similar to Section \ref{sec:ase} for ASE, an example embedding can be obtained using the SVD of $\tilde{A}$, which can then be used to formulate a graph clustering procedure, as described in the next section.  

\subsection{Proposed Graph Clustering Algorithm (\dasec)}\label{sec:daseclust}

The normalised doubled adjacency matrix ($AA/N$) and the unnormalised doubled adjacency matrix ($AA$) yield equivalent spectral embeddings, as normalisation simply rescales the singular values. Consequently, without loss of generality, we use the unnormalised version, which corresponds to a weighted SBM with Poisson-Binomial edge weights.

Our proposed clustering algorithm is based on forming a spectral embedding of the unnormalised doubled adjacency matrix. Whilst similar in nature to the clustering using ASE (Algorithm \ref{alg:ASE}), this subtle change in embedding matrix has significant consequences. More specifically, we show in Section \ref{sec:theory} that the spectral embedding of our unnormalised doubled adjacency matrix provides a tighter concentration bound for misclustering rates than that of the unnormalised adjacency matrix, in terms of the Frobenius norm metric. We then provide a consistent node cluster assignment via $k$-means or a Gaussian mixture model applied to our embedding. The steps of the proposed algorithm (\dasec) are summarised below in Algorithm \ref{alg:DASE}.

\begin{algorithm}[H]
\caption{Doubled Adjacency Spectral Embedding (DASE)-CLUST}
\label{alg:DASE}
\hspace*{\algorithmicindent} 
\textbf{Input}: Adjacency matrix $A \in \{ 0, 1 \}^{N \times N}$, number of clusters $K \in \mathbb{N}$ and latent dimension $d \in \mathbb{N}$\\
 \hspace*{\algorithmicindent} 
 \textbf{Output}: Clustering labels $l \in \mathbb{R}^N$
\begin{algorithmic}[1]
\State Compute the doubled adjacency matrix $\tilde{A}$
\State Compute the singular value decomposition (SVD) of $\tilde{A}$, i.e., $\tilde{A} = \tilde{U} \tilde{\Sigma} \tilde{V}^T$, where the singular values in $\tilde{\Sigma}$ are ordered in decreasing magnitude.
\If{$A$ is directed}
    \State Define the embedding $S = \left[ \tilde{U}_d\tilde{\Sigma}_d^{1/2}| \tilde{V}_d\tilde{\Sigma}_d^{1/2} \right] \in \mathbb{R}^{N \times 2d}$, where $\tilde{U}_d$ and $\tilde{V}_d$ contain the top $d$ left and right singular vectors, respectively. $\tilde{\Sigma}_d$ contains the corresponding singular values.
\Else
    \State Define the embedding $S = \tilde{U}_d \tilde{\Sigma}_d^{1/2} \in \mathbb{R}^{N \times d}$, where $\tilde{U}_d$ contains the top $d$ eigenvectors corresponding to the largest $d$ singular values in $\tilde{\Sigma}_d$.
\EndIf
\State Extract the clustering labels $l$ from either $k$-means or a Gaussian mixture model (with $K$ clusters) applied to the rows of $S$. 
\end{algorithmic}
\end{algorithm}

\subsection{Theoretical Properties of the Proposed \dasec\ Algorithm}
\label{sec:theory}

We now provide theoretical justification for our proposed \dasec\ approach, focusing on its consistency properties. The proofs of the results in this section can be found in Appendix \ref{appendix:proofs}. 

\subsubsection{Notation and Assumptions}
\label{sec:notations}
We use the following notation for the remainder of this article. Let $d = \text{rank}(B)$, where $B$ is the block probability matrix defined in Definition \ref{def:sbm}. From Section \ref{sec:sbm}, recall that $Q = \mathbb{E}(A) =  ZBZ^\top$, where the block membership indicator matrix $Z \in \{0,1\}^{N \times K}$. Let $Z^\top Z =: D_n = \textbf{diag}(n_1, \dots, n_K) \in \mathbb{R}^{K \times K}$, where $n_i=\#\{u : \theta(u) = i\}$ is the size of group $i$ counted via the block assignment function $\theta:[N]\rightarrow [K]$, and $\textbf{diag}(n_1, \dots, n_K)$ denotes the diagonal matrix such that $(D_n)_{ii}=n_i$ for $i=1,\dots,K$.

We denote $\tilde{B} := B D_n B$. Let the singular value decomposition of $\tilde{B}$ be
    \[
    \tilde{B} = \tilde{L} \tilde{\Lambda} \tilde{R}^\top,
    \]
where $\tilde{L}, \tilde{R} \in \mathbb{R}^{K \times d}$ and $\tilde{\Lambda} \in \mathbb{R}^{d \times d}$ is diagonal. We assume that $\lambda_d(\tilde{B}) = \tilde{b} N$ for some constant $\tilde{b} > 0$. Suppose $\tilde{\pi} = \min_{i \in \{1, \dots, K\} } \pi_i$, where $\pi_i = n_i / N \in (0,1)$ denotes the proportion of nodes in the group $i$. Here, $\tilde{b}$ and $\tilde{\pi}$ are not dependent on $N$.\\

The recovery of cluster memberships for each node is dependent on the quality of the embedding $S$ in Algorithm \ref{alg:DASE}. If the embedding encodes the distribution of the nodes in different clusters distinctly, it renders the clustering algorithm to be consistent in finding the communities. 
%
%
%
The main result of this section establishes the consistency of \dasec, in terms of the estimation of the community membership for each node based on the block assignment function, which assigns blocks based on $S$. \\

To establish this result, we make the following assumptions, in a similar vein to \cite{sussman2012consistent}.

\begin{assumption}
\label{assum:1}
Assume a sequence of random adjacency matrices \(A^{(N)}\) with node set \([N]\) for $N\in\{1,2,\ldots,\}$. The edges are distributed according to a stochastic block model (see Definition \ref{def:sbm}) with parameters \(B\) and \(\pi\), with \(B \in [0,1]^{K \times K}\) and \(\pi \in (0,1)^{K}\) be a vector with positive entries summing to unity; let $Q^{(N)} = \mathbb{E}(A^{(N)})$ be a sequence of matrices representing the probability of connection between nodes for a fixed node set $[N]$.
\end{assumption}

\begin{assumption}
\label{assum:2}
Assume \(\mathrm{rank}(B) = d\).
\end{assumption}

\begin{assumption}
\label{assum:3}
   Suppose $\tilde{B} = \tilde{\mathbf{\nu}} \tilde{\mathbf{\mu}}^\top$, where $\tilde{\mathbf{\nu}}, \tilde{\mathbf{\mu}} \in \mathbb{R}^{K \times d}$. Define $\tilde{\beta} > 0$ such that 
    \[
    \tilde{\beta} < \lVert \tilde{\nu}_u - \tilde{\nu}_v \rVert \qquad \text{or} \qquad \tilde{\beta} < \lVert \tilde{\mu}_u - \tilde{\mu}_v \rVert \quad \text{for all} \quad u \ne v\in[N].
    \]
    Moreover, we write $\tilde{\beta} = \hat{\beta}N$ for some constant $\hat{\beta} > 0$. 
\end{assumption}

We denote by $\theta^{(N)} : [N] \rightarrow [K]$ a random block membership function corresponding to the adjacency $A^{(N)}$.

\subsubsection{Consistency of Clustering Using \dasec}\label{sec:consist-gen}

In this section, we present consistency results for both general networks generated under the SBM, and the special case of core-periphery networks. Our approach to establish consistency of misclustering using DASE follows a similar strategy to \cite{sussman2012consistent}, in which we first consider the properties of $\tilde{A}\tilde{A}^\top$, from which results from spectral theory allow us to quantify the behaviour of the singular values of $\tilde{A}$ and therefore the embedding. This is then used to quantify the misclustering rate using DASE within a clustering algorithm.  
%

For the results in this section, we initially focus on the directed network setting; we provide a treatment of undirected networks as a corollary to the main  consistency result.




\begin{theorem}
\label{thm_dir}
    Assume that the graph under consideration is directed, and that Assumptions \ref{assum:1}, \ref{assum:2} and \ref{assum:3} in Section \ref{sec:notations} hold. Suppose also that the number of blocks $K$ and the latent vector dimension $d$ are known. Let $\hat{\theta}^{(N)}: \mathcal{V} \mapsto \{ 1, \dots, K \}$ be the estimated block assignment obtained by clustering the rows of $\tilde{W}^{(N)}=[\tilde{U}^{(N)}|\tilde{V}^{(N)}]$, satisfying Equation~(\ref{eq:clust}) for each $N$. Let $H_K$ denote the set of permutations of $[K]$. Then, with probability at least $1-4N^{-2}$, 
    \begin{equation*}
        \min_{\rho \in H_K} | \{ u \in \mathcal{V} : \theta^{(N)}(u) \ne \rho(\hat{\theta}^{(N)}(u)) \} | \leq \frac{2^4 \, 3^2}{\hat{\beta}^2 (\tilde{b} \tilde{\pi})^5~} \cdot \frac{(N-1)^4}{N^4} \cdot \frac{\log N}{N^2}.
    \end{equation*}
\end{theorem}
\bigskip
The proof of Theorem \ref{thm_dir} is given in Appendix \ref{sec:proofthm1}.\\


For undirected graphs, we can exploit the equality of left and right singular values to derive the misclustering rate is as follows.

\begin{corollary}\label{corrl_undirected}
    Under the undirected condition in Section \ref{sec:notations} and Assumptions \ref{assum:1}, \ref{assum:2} and \ref{assum:3}, suppose that the number of blocks $K$ and the latent vector dimension $d$ are known. Let $\hat{\theta}^{(N)}: \mathcal{V} \mapsto \{ 1, \dots, K \}$ be the estimated block assignment obtained by clustering the rows of $\tilde{W}^{(N)}$ where $\tilde{W}^{(N)}=\tilde{U}^{(N)} (\tilde{\Sigma}^{(N)})^{1/2}$, satisfying Equation~(\ref{eq:clust}) for each $N$. Let $H_K$ denote the set of permutations of $[K]$. With probability at least $1-4N^{-2}$, 
    \begin{equation*}
        \min_{\rho \in H_K} | \{ u \in \mathcal{V} : \theta^{(N)}(u) \ne \rho(\hat{\theta}^{(N)}(u)) \} | \leq \frac{2^3 \, 3^2}{\hat{\beta}^2 (\tilde{b} \tilde{\pi})^5~} \cdot \frac{(N-1)^4}{N^4} \cdot \frac{\log N}{N^2}.
    \end{equation*}
\end{corollary}
\bigskip
The proof of this result is given in Appendix \ref{sec:proofcorrl_undirected}, where it is noted that the difference in the bounds between the directed and undirected network cases is just a factor of $2$ resulting from the equality of left and right singular vectors in the embedding of the doubled adjacency matrix.

\subsubsection{Consistency Results for Core-periphery Graphs}\label{sec:cpclust}
We now consider the rate of consistency in the specific case of core-periphery graphs, whose structure has been shown to represented well with adjacency-based spectral embedding.  To our knowledge, this is the first work which provides a theoretical result of rate of misclassification specific to core-periphery networks, albeit with $K=2$. 

To derive the misclustering rate on core-periphery graphs, we modify Assumption \ref{assum:2} in the following way.

\begin{assumption}
\label{assum:4}
Assume that \(\mathrm{rank}(B) = 2\) and that $B$ represents connection probability matrix in a core-periphery network as defined in Definition \ref{def:core-periphery}. Let the within-core and within-periphery connection probabilities be $p$ and $s$, respectively, and denote the between group connection probabilities by $q$ and $r$. Assume the ordering $s<r,q<p$.
\end{assumption}

\bigskip

Note that under Assumption \ref{assum:4}, the matrix $Q$ can be written
\begin{equation}
\label{eq:Q}
    Q = 
    \begin{bmatrix}
    p J_{n_1 \times n_1} & q J_{n_1 \times n_2} \\
    r J_{n_2 \times n_1} & s J_{n_2 \times n_2}
    \end{bmatrix},
\end{equation}
where $J_{n,m}$ is a matrix $n \times m$ whose entries are all equal to $1$.\\

With these conditions, we establish the following misclustering rate of \dasec\ when the graph is directed and has a core-periphery structure.

\begin{theorem}
\label{thm_dir_core}
    Under the directed condition in Section \ref{sec:notations} and Assumptions \ref{assum:1}, \ref{assum:3} and \ref{assum:4}, suppose that the number of blocks $K$ and the latent vector dimension $d$ are known. Let $\hat{\theta}^{(N)}: \mathcal{V} \mapsto \{ 1, \dots, K \}$ be the block assignment obtained by clustering the rows of $\tilde{W}^{(N)}$, where $\tilde{W}^{(N)}=[\tilde{U}^{(N)}|\tilde{V}^{(N)}]$, satisfying Equation~(\ref{eq:clust}) for each $N$. Let $H_K$ denote the set of permutations of $[K]$. Then, it almost always holds that
    \begin{equation*}
        \min_{\rho \in H_K} | \{ u \in \mathcal{V} : \theta^{(N)}(u) \ne \rho(\hat{\theta}^{(N)}(u)) \} | \leq \frac{2^4 \, 3^2 \, (p^4 - s^4)^2}{\hat{\beta}^2 (\tilde{b}\tilde{\pi})^5} \cdot \frac{(N-1)^4}{N^4} \cdot \frac{\log N}{N^2}.
    \end{equation*}
\end{theorem}
\bigskip
We refer the reader to Appendix \ref{sec:proofthmcore} for the proof of Theorem \ref{thm_dir_core}. Note that the only difference between the bound achieved for the core-periphery case in Theorem \ref{thm_dir_core} and that of the general network case in Theorem \ref{thm_dir} is a factor of $(p^4-s^4)^2$.  This highlights that the misclustering rate is dependent on the core and periphery cluster connection probabilities, with better (theoretical) clustering achieved if there is a strongly distinct core compared to that in the periphery.\\

For the purpose of comparing clustering consistency between DASE and ASE, we also derive the misclustering rate for ASE in the core–periphery graph case based on the proof for general graphs in \cite{sussman2012consistent}. To this end, similar to Section \ref{sec:notations} let the singular value decomposition of $B$ be denoted
\[
B = L \Lambda R^T,
\]
where $L, R \in \mathbb{R}^{K \times d}$ and $\Lambda \in \mathbb{R}^{d \times d}$ is diagonal, and again we assume $\lambda_d(B) = b$ for some constant $b>0$, with $b$ not dependent on $N$.  The equivalent of Assumption \ref{assum:3} for ASE is\\

\begin{assumption}
\label{assum:5}
   Suppose $B = \mathbf{\nu} \mathbf{\mu}^\top$, where $\mathbf{\nu}, \mathbf{\mu} \in \mathbb{R}^{K \times d}$. Define $\beta > 0$ such that 
    \[
   \beta < \lVert \nu_u - \nu_v \rVert \qquad \text{or} \qquad \beta < \lVert \mu_u - \mu_v \rVert \quad \text{for all} \quad u \ne v\in[N].
    \]
\end{assumption}

The consistency bound for \asec\ for core–periphery graphs is as follows.

\begin{corollary}
    \label{corrl_core_ase}
    Under the directed condition in Section \ref{sec:notations} and Assumptions \ref{assum:1}, \ref{assum:4} and \ref{assum:5}, suppose that the number of blocks $K$ and the latent vector dimension $d$ are known. Let $\hat{\theta}^{(N)}: \mathcal{V} \mapsto \{ 1, \dots, K \}$ be the block assignment obtained by clustering the rows of $W$, where $W^{(N)}=[U^{(N)}|V^{(N)}]$, and $U^{(N)}$ and $V^{(N)}$ as defined in Algorithm \ref{alg:ASE} (for fixed $N$), satisfying Equation~(\ref{eq:clust}) for each $N$. Let $H_K$ denote the set of permutations of $[K]$. Then, it almost always holds that
    \begin{equation*}
        \min_{\rho \in H_K} | \{ u \in \mathcal{V} : \theta^{(N)}(u) \ne \rho(\hat{\theta}^{(N)}(u)) \} | \leq \frac{2^4 \, 3^3 \, (1-s^2+p^2)^2}{\beta^2 (b\tilde{\pi})^5} \, \log N.
    \end{equation*}
\end{corollary}
\bigskip
The proof of Corollary \ref{corrl_core_ase} is provided in Appendix \ref{appendix:core_ase}.\\

From Theorem \ref{thm_dir}, it can be seen that \dasec\ achieves a tighter bound than the bound of $\frac{2^4 \, 3^3}{\beta^2 (b\tilde{\pi})^5} \, \log N$ for \asec\ reported in \cite{sussman2012consistent}. This follows through to the core-periphery case based on Theorem \ref{thm_dir_core} and Corollary \ref{corrl_core_ase}, where it can also be seen that as $N$ increases, the misclustering rate of \dasec\ converges to zero faster than that of \asec, indicating that graph clustering using our proposed doubled adjacency-based embedding achieves a stronger consistency guarantee in core-periphery networks.  

Note that consistency results for {\em undirected} core-periphery graphs can be obtained similar to the general case in Section \ref{sec:consist-gen}, and the bounds in Theorem \ref{thm_dir_core} and Corollary \ref{corrl_core_ase} would change by a multiplicative factor of $2$ as in the general network case (Corollary \ref{corrl_undirected}).

\section{Simulation Study}
\label{sec:sims}
In this section, we empirically evaluate our proposed DASE method, and compare it to spectral clustering \parencite{ng2001spectral} and Adjacency Spectral Embedding (ASE) \parencite{sussman2012consistent} using simulated data. To ensure generality, we consider both directed and undirected graphs. Similar to most existing studies \parencite{priebe2019two, sussman2012consistent}, we consider networks with only two clusters.

For each setting, we examine three simulation scenarios: (i) fixing the network size, $N$, while varying the network density, $\alpha$ (see below) (ii) fixing the network density while varying the network size, and (iii) fixing both the network size and network density while varying the group proportions, $\pi$. We define the network (edge) density as
\[ \alpha = 
\begin{cases}
\dfrac{\lvert \mathcal{E} \rvert}{N(N-1)}, & \text{if the network is directed}, \\[8pt] \dfrac{2\lvert \mathcal{E} \rvert}{N(N-1)}, & \text{if the network is undirected}, 
\end{cases}
\]
where $\lvert \mathcal{E} \rvert$ denotes the number of edges in the network. We generate simulated adjacency matrices using the stochastic block model \eqref{eq:sbm}  using the block probability $B$ specified by
\begin{equation*}
    B = 
\begin{bmatrix}
    b_{11} & b_{12} \\
    b_{21} & b_{22}
\end{bmatrix} = \alpha \times \pi \times
\begin{bmatrix}
    r_{11} & r_{12} \\
    r_{21} & r_{22}
\end{bmatrix},
\end{equation*}
where $\alpha$ controls the expected network density and $r_{ij}$ represents the relative proportion of edges between groups $i$ and $j$. We set $r_{11} = 1$ as the baseline ratio. This structure ensures that the relative proportions of within-group and between-group edges remain consistent when varying the overall density (via $\alpha$) while keeping the network size fixed. For example, when $K=2$, if there are $100$ edges in the first group, then there are approximately $60$ edges between two groups and $30$ edges in the second group. Thus, this proportional matrix allows us to vary the density without altering the underlying connectivity structure, hence we adopt this configuration throughout our simulation study, and focus on sparse (low density) graphs.  

We use the Normalised Mutual Information (NMI) \parencite{strehl2002cluster} to assess clustering performance. NMI measures the common indices between two cluster labellings. It ranges between $0$ and $1$, where $0$ indicates no mutual information and $1$ indicates identical labels. For each simulation setting, we obtain the average NMI and its corresponding standard deviation over $n_{rep}=50$ simulated graphs using $k$-means. All simulations were implemented using \texttt{Python} version 3.9.7 \parencite{python397}, using packages \texttt{Numpy} \parencite{harris2020array}, \texttt{Networkx} \parencite{hagberg2008networkx}, \texttt{Scikit-learn} \parencite{pedregosa2011scikit} and \texttt{SciPy} \parencite{2020SciPy-NMeth}. For the second (respectively third) step of Algorithm \ref{alg:ASE} (respectively Algorithm \ref{alg:DASE}), we need to compute the $d$ largest singular values and their corresponding singular vectors. However, for large networks, calculating all $N$ singular values and vectors is computationally expensive; for computational efficiency and consistency across clustering methods, we directly compute only the $d$ largest singular values and vectors using \texttt{svds} from the \texttt{scipy.sparse.linalg} package \parencite{2020SciPy-NMeth}. Similarly, for the benchmark spectral clustering algorithm, only the first $d$ eigenvalues and eigenvectors are computed. 


Simulation results for directed and undirected graphs are presented in Sections \ref{subsec:sim_dir} and \ref{subsec:sim_undir}, respectively. For both settings, the performance of the clustering methods with $k$-means are represented in figures as follows: spectral clustering (circle, blue); \asec\ (triangle, orange); \dasec\ (square, green). 

Code to implement the proposed \dasec\ algorithm using \texttt{Python}, including code for the simulation study and analysis of the real datasets in Section \ref{sec:data}, is available at \url{https://github.com/SinyoungPark2681/dase}.

\subsection{Directed Graphs}
\label{subsec:sim_dir}
For our first experiment to examine the effect of network density, we use a balanced two-cluster scenario, i.e., $\pi = (0.5, 0.5)^\top$, and set the block probability matrix $B$ as:
\[
B = \alpha \times \pi \times
    \begin{bmatrix}
        1 & 0.6 \\
        0.6 & 0.3
    \end{bmatrix}.
\]

The clustering performance on directed graphs using $k$-means is shown in Figure \ref{fig:Figure_2}. Figure \ref{fig:Fig2a} illustrates the average NMI with the corresponding standard deviation (shaded areas) over $n_{rep}=50$ simulated graphs for a fixed network size of $N=1,000$, while varying the network density with $\alpha \in (0.03, 0.2)$. Figure \ref{fig:Fig2b} illustrates the average NMI and corresponding NMI standard deviation (shaded areas) for a fixed expected network density of $\alpha=0.05$ while varying the network size. 

\begin{figure*}[!h]
    \centering
    \begin{subfigure}{0.4\textwidth}
        \centering
        \includegraphics[width=\linewidth]{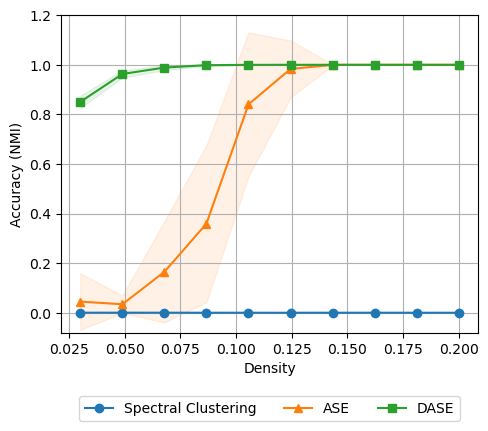}
        \caption{NMI with fixed network size}
        \label{fig:Fig2a}
    \end{subfigure}
    \begin{subfigure}{0.4\textwidth}
        \centering
        \includegraphics[width=\linewidth]{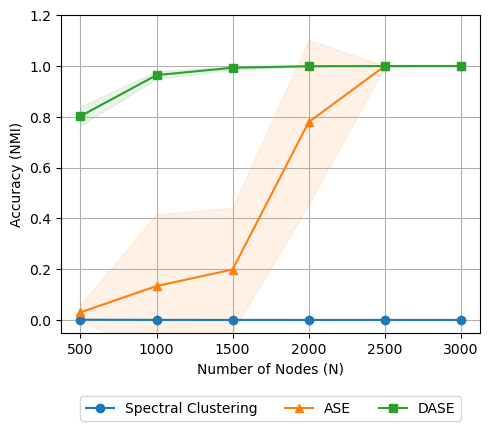}
        \caption{NMI with fixed network density}
        \label{fig:Fig2b}
    \end{subfigure}

    \caption{Comparison of clustering performance on directed graphs using $k$-means in terms of mean accuracy (NMI) over $n_{rep}=50$ simulated graphs when $K=2$ with $\pi = (0.5, 0.5)^\top$: (a) NMI with fixed network size ($N=1,000$) and varying network density; (b) NMI with fixed expected network density ($\alpha = 0.05$) and varying network sizes. In (a) and (b), shaded areas represent the standard deviations of the NMI values over the $n_{rep}$ simulated graphs.}
    \label{fig:Figure_2}
\end{figure*}

As seen from Figures \ref{fig:Fig2a} and \ref{fig:Fig2b}, spectral clustering fails completely to capture the underlying network structure. In contrast, \asec\ and \dasec, which are more suitable for the core-periphery graphs, are able to successfully recover the community structure in this setting. However, when the network is sparse or the network is small, \asec\ not only struggles to capture the network structure but also exhibits clustering instability, as indicated by its larger standard deviation (shaded areas) compared with that of \dasec.

\begin{figure}[!h]
    \centering
    \includegraphics[width=0.4\linewidth]{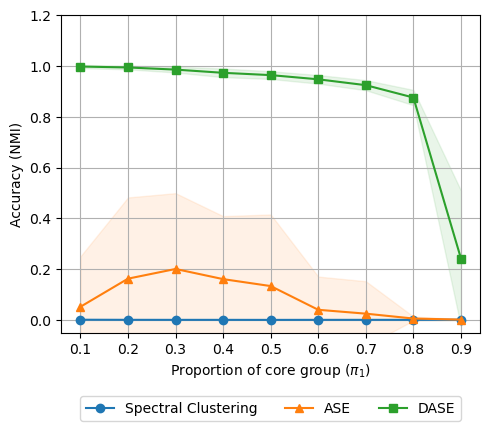}
    \caption{Comparison of clustering performance on directed graphs using $k$-means in terms of mean NMI (line) and the corresponding standard deviations (shaded areas) over $n_{rep}=50$ simulated graphs when $K=2$. In the simulation, the network size is fixed at $(N=1,000)$, and the block probability matrix $B$ is fixed, while varying the core group ratio ($\pi_1$) from $0.1$ to $0.9$.}
    \label{fig:Fig3}
\end{figure}

We examine clustering performance for unbalanced two-cluster networks in Figure \ref{fig:Fig3}. In this setting, the network size is fixed at $N=1,000$, the network density is fixed at $\alpha = 0.05$, and the relative edge-probability ratio $r_{ij}$s are fixed by
\[
\begin{bmatrix}
    r_{11} & r_{12} \\
    r_{21} & r_{22}
\end{bmatrix}
=
\begin{bmatrix}
    1 & 0.6 \\
    0.6 & 0.3
\end{bmatrix},
\]
while the group proportions vary, such as $\pi = (0.1, 0.9)^\top$ and $\pi = (0.9, 0.1)^\top$. Since the network density is fixed, it remains constant at $\alpha = 0.05$ across all group proportions. Here, we focus on the proportion of the size of the core group ($\pi_1$). As in the previous simulations, the average NMI and the corresponding standard deviation are computed over $n_{rep}=50$ simulated graphs.


As expected, spectral clustering fails to detect the network structures across different core sizes. For \asec, the accuracy is highest when the network is unbalanced ($\pi = (0.3, 0.7)^\top$). For both \asec\ and \dasec, the accuracy decreases as the core group becomes more dominant. This occurs because, when the core group is large, the network increasingly resembles a single dense cluster, making it challenging to distinguish between the core and periphery. In this case, \dasec\ cannot detect the structure when $90\%$ of the nodes belong to the core group ($\pi = (0.9, 0.1)^\top$), whereas \asec\ struggles to detect the structure, particularly when the core group is dominant.

\subsection{Undirected Graphs}
\label{subsec:sim_undir}
We now consider undirected, balanced two-cluster networks $( \pi = (0.5, 0.5)^\top )$ to evaluate the clustering performance of spectral clustering, \asec, and \dasec\ when the network density changes, shown in Figure \ref{fig:Fig5}. Similar to the directed case, our primary comparisons are based on clustering accuracy (average NMI) and the corresponding standard deviations (shaded areas). Figures \ref{fig:Fig5a} and \ref{fig:Fig5b} illustrate the average NMI with the corresponding standard deviation over $n_{rep}=50$ simulated graphs for a fixed network size of $N=1,000$ and a fixed expected network density of $\alpha=0.05$, respectively (with the same specifications as Section \ref{subsec:sim_dir}).

\begin{figure}[!h]
    \centering
    \begin{subfigure}{0.4\textwidth}
        \centering
        \includegraphics[width=\linewidth]{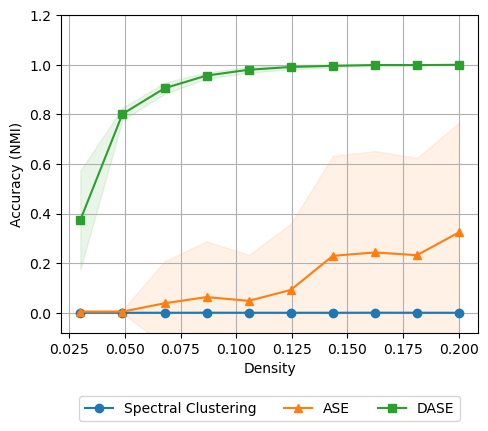}
        \caption{NMI with fixed network size}
        \label{fig:Fig5a}
    \end{subfigure}
    \begin{subfigure}{0.4\textwidth}
        \centering
        \includegraphics[width=\linewidth]{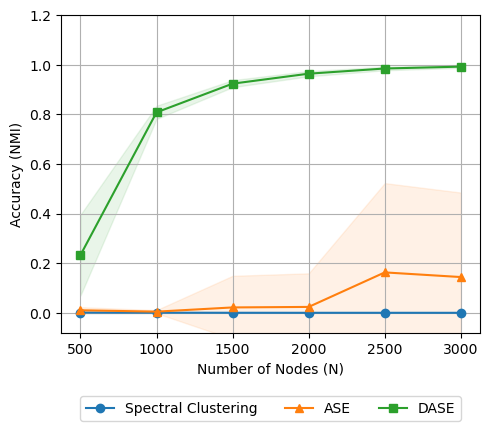}
        \caption{NMI with fixed network density}
        \label{fig:Fig5b}
    \end{subfigure}   

    \caption{Comparison of clustering performance on undirected graphs using $k$-means in terms of mean accuracy (NMI) over $n_{rep}=50$ simulated graphs when $K=2$ with $\pi = (0.5, 0.5)^\top$: (a) NMI with fixed network size ($N=1,000$) and varying network density; (b) NMI with fixed expected network density ($\alpha = 0.05$) and varying network sizes. In (a) and (b), shaded areas represent the standard deviations of the NMI values over the $n_{rep}$ simulated graphs.}
    \label{fig:Fig5}
\end{figure}

Figure \ref{fig:Fig5a} illustrates that \asec\ performs better than spectral clustering, as expected for core-periphery graphs. However, as the network becomes denser, \asec\ still struggles to identify the underlying network structures and exhibits large variability. In contrast, \dasec\ outperforms both methods, even when the network is sparse. When the network density is fixed (Figure \ref{fig:Fig5b}), ASE fails to recover accurate clustering assignments even as the network size increases, and it continues to yield a larger standard deviation. In comparison, \dasec\ converges more rapidly to the correct clustering and maintains a smaller standard deviation across network sizes.

Figure \ref{fig:Fig6} shows performance of the clustering methods for unbalanced two-cluster undirected networks. Similar to Section \ref{subsec:sim_dir}, for these simulations, the network size is fixed at $N=1,000$ and the network density is fixed at $\alpha = 0.05$, and the relative edge-probability ratios are fixed by $\begin{bmatrix}
    r_{11} & r_{12} \\
    r_{21} & r_{22}
\end{bmatrix}
=
\begin{bmatrix}
    1 & 0.6 \\
    0.6 & 0.3
\end{bmatrix}$,
while the proportional size of the core cluster, $\pi_1$, is varied.



In this setting, \asec\ consistently struggles to identify the network structure across all core cluster proportions. In contrast, \dasec\ achieves substantially better performance over a wide range of core size proportions, except $\pi_1 \geq 0.7$. Even in these extreme cases ($\pi_1 = 0.1$ and $\pi_1 \geq 0.8$), \dasec\ still outperforms \asec, although the overall accuracy decreases for both methods. Overall, we conclude that \dasec\ achieves better clustering performance in challenging settings where the core group forms the majority cluster, making the community structure less distinct.

\begin{figure}[!ht]
    \centering
    \includegraphics[width=0.4\linewidth]{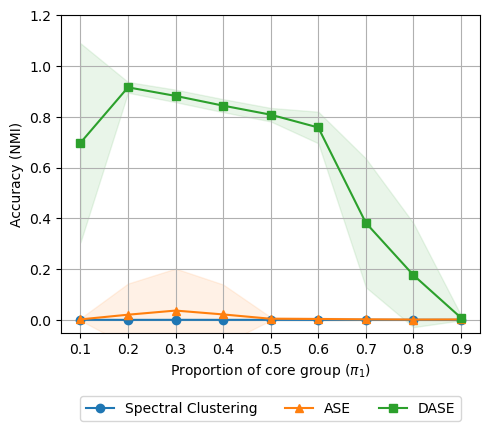}
    \caption{Comparison of clustering performance on undirected graphs using $k$-means in terms of mean NMI (line) and the corresponding standard deviations (shaded areas) over $n_{rep}=50$ simulated graphs when $K=2$. In the simulation, the network size is fixed at $(N=1,000)$, and the block probability matrix $B$ is fixed, while varying the core group ratio ($\pi_1$) from $0.1$ to $0.9$.}
    \label{fig:Fig6}
\end{figure}

Contrasting the results from Sections \ref{subsec:sim_dir} and \ref{subsec:sim_undir}, we observe a clear difference in clustering performance between directed and undirected graphs. This is expected, as directed graphs retain additional information through edge directions, while undirected graphs only preserve the presence or absence of connections. Overall, whilst \asec\ performs noticeably worse on undirected graphs than on directed graphs, \dasec\ provides superior clustering performance, particularly when the network is sparse or the network size is small, for both directed and undirected settings, and has much lower variability in clustering performance.

Additional results using the Gaussian Mixture Model (GMM) are shown in Appendix \ref{appendix:sims}, where we also consider directed and undirected graphs. In this setting, \asec\ and \dasec\ recover the community structure. However, \asec\ exhibits clustering instability as it has a larger standard deviation compared to \dasec.

\subsection{Assessing the DASE Embedding with Chernoff Information}\label{sec:chernoff}
The Chernoff information has been widely used as a metric for the separability between two distributions \parencite{chernoff1952measure, chernoff1956large}. Based on this measure, several works have proposed using Chernoff information to quantify community separation in spectral embedding spaces, see e.g., \cite{tang2018limit, priebe2019two, gallagher2024spectral}. A large Chernoff information indicates that two embedded distributions are more statistically distinguishable, and thus provides a more informative representation with which to perform graph clustering.  In this section, we investigate the size-adjusted Chernoff information associated to DASE, comparing it to that of ASE (this measure modifies the Chernoff information so that it is independent of the network size).  We first state a result from \cite{gallagher2024spectral} which simplifies calculation of the (size-adjusted) Chernoff information for weighted symmetric graphs, which is useful for characterising embeddings under an asymptotic Gaussian mixture representation, and will serve to aid our comparative discussion. 

\begin{lemma}[\cite{gallagher2024spectral}.]
    Let $K$ be the number of communities in the network. Suppose a graph connectivity matrix of interest is generated from a (weighted) stochastic block model as defined in \eqref{eq:wrdpg} with full-rank mean matrix $M=\mathbb{E}(G)$, $\text{rank}(M) = K$. Then the size-adjusted Chernoff information is given by
    \begin{equation}
        CI = \min_{k \ne l} \sup_{t \in (0,1)} \left[ \frac{t(1-t)}{2} \left\{ e^T M \Pi S_{kl}(t)^{-1} M e \right\} \right],
        \label{eq:Chernoff_Information}
    \end{equation}
    where $e=(e_k - e_l)$, $e_k \in \mathbb{R}^K$ is the standard basis vector with $k$-th element as one, zero elsewhere, $\Pi = \textbf{diag}(\pi_1, \dots, \pi_K)$, and
    \[
    S_{kl}(t) = (1-t)\textbf{diag}(C_k) + t \textbf{diag}(C_l) \in \mathbb{R}^{K \times K}
    \]
    is a convex combination of the diagonal matrices formed from rows $C_k$ and $C_l$ of the block variance matrix $V$.
\end{lemma}
\bigskip

We now compare the Chernoff information for ASE and DASE using \eqref{eq:Chernoff_Information} from the above lemma. For ASE, the mean block matrix $M$ and the block variance matrix $C$ are defined as
\begin{align*}
    M_{\theta_i, \theta_j} &= B_{\theta_i, \theta_j} \in [0,1] \quad \text{and} \\
    C_{\theta_i, \theta_j} &= (1-M_{\theta_i, \theta_j}) M_{\theta_i, \theta_j} \in [0,1].
\end{align*}
where $\theta_i, \theta_j$ are the block labels of nodes $i$ and $j$, respectively.

As noted in Section \ref{sec:dase}, the doubled adjacency $\tilde{A}$ follows a Poisson-Binomial distribution.  In what follows, we normalise $\tilde{A}$ by $N$ so that its entries have bounded expectation (independent of $N$); its entries then  satisfy Assumption 1(ii) of \cite{gallagher2024spectral}, which facilitates the derivation of the Chernoff information for the weighted generalized random dot product model using the associated asymptotic Gaussian mixture model representation. The entries of the normalised doubled adjacency thus have mean  
\begin{align*}
    \tilde{Q}_{ij} &= \frac{1}{N} \sum_{k=1}^N Q_{ik}Q_{kj} \in [0,1] 
\end{align*}
and variance
\begin{align*}
    \tilde{Q}'_{ij} &= \frac{1}{N^2} \sum_{k=1}^N (1-Q_{ik}Q_{kj})Q_{ik}Q_{kj} \in [0,1].
\end{align*}

Then the mean block matrix $\tilde{M}$ and the block variance matrix $\tilde{C}$ for DASE are bounded and defined as
\begin{align*}
    \tilde{M}_{\theta_{i} \theta_{j}} &= \frac{1}{n_{\theta_i} n_{\theta_j}} \sum_{a \in G_{\theta_i}} \sum_{b \in G_{\theta_j}} \tilde{Q}_{ab} \quad \text{and} \\
    \tilde{C}_{\theta_{i} \theta_{j}} &= \frac{1}{n_{\theta_i} n_{\theta_j}} \sum_{a \in G_{\theta_i}} \sum_{b \in G_{\theta_j}} \tilde{Q}'_{ab},
\end{align*}
where $G_{\theta_i}$ denotes the $\theta_i$-th group, and $n_{\theta_i} = |G_{\theta_i}|$ denotes the size of the $\theta_i$-th group.\\




To empirically compare the Chernoff information of ASE and DASE, we conduct two simulation studies under the stochastic block model. In the first simulation, we use the probability matrix
\[
B = \alpha \times \pi \times
\begin{bmatrix}
    1 & 0.6 \\
    0.6 & 0.3
\end{bmatrix},
\]
where as before, $\alpha$ controls the overall network density. The number of nodes is fixed at $N=1,000$ and the group proportion is fixed at $\pi=(0.5,0.5)^\top$. We examine two regimes of network density: (i) an overall density range $\alpha \in [0.05,0.6]$ and (ii) a sparse density range $\alpha \in [0.03,0.2]$. For each value of $\alpha$, the Chernoff information of ASE and DASE is averaged over $n_{\mathrm{rep}}=50$ independent network realisations.

In the second simulation, we use the same probability matrix structure as follows:
\[
B = \alpha \times \pi \times
\begin{bmatrix}
    1 & 0.6 \\
    0.6 & 0.3
\end{bmatrix} = 
\begin{bmatrix}
    p & q \\ r & s
\end{bmatrix},
\]

with the number of nodes fixed at $N=1,000$ and the network density fixed at $\alpha=0.3$, while the group proportion $\pi=(\pi_1,\pi_2)^\top$ is varied.

To maintain a fixed overall density $\alpha$ as $\pi$ varies, the block probabilities are rescaled accordingly, so that the resulting probability matrix depends on $\pi$. For example, when $\pi=(0.1,0.9)^\top$ and $\pi=(0.9,0.1)^\top$, the corresponding probability matrices become
\[
B_1 =
\begin{bmatrix}
    0.83 & 0.50 \\
    0.50 & 0.25
\end{bmatrix}
\qquad \text{and} \qquad
B_2 =
\begin{bmatrix}
    0.33 & 0.20 \\
    0.20 & 0.10
\end{bmatrix},
\]
respectively. When the periphery group dominates (e.g., $\pi_2=0.9$), achieving the same target density requires a substantially higher within-core connection probability.



\begin{figure}[!h]
    \centering
    \begin{subfigure}[b]{0.32\textwidth}
        \centering
        \includegraphics[scale=0.41]{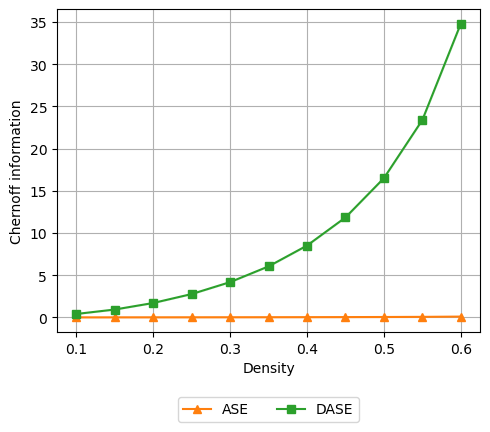}
        \caption{Overall network density ($\alpha_1$)}
        \label{fig:chernoff_density_overall}
    \end{subfigure}
    \begin{subfigure}[b]{0.32\textwidth}
        \centering
        \includegraphics[scale=0.41]{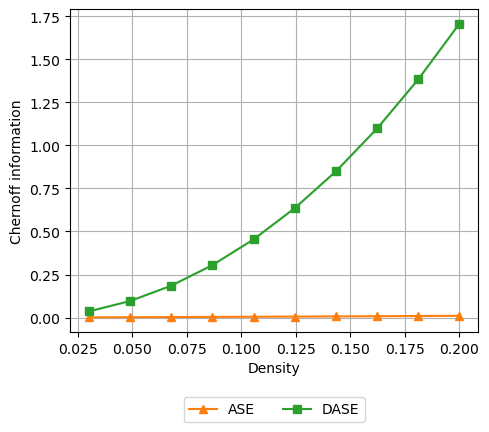}
        \caption{Sparse network density ($\alpha_2$)}
        \label{fig:chernoff_density}
    \end{subfigure}
    \begin{subfigure}[b]{0.32\textwidth}
        \centering
        \includegraphics[scale=0.41]{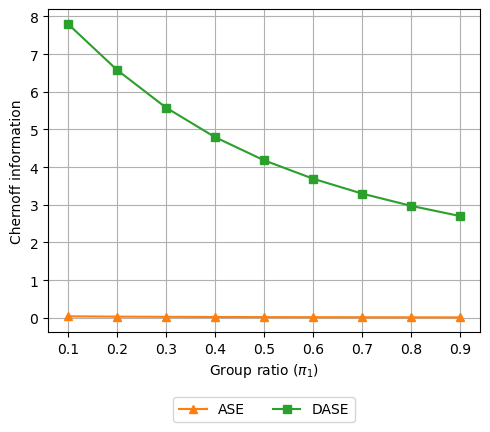}
        \caption{Varying the group ratio $(\pi_1, \pi_2)$}
        \label{fig:cheroff_ratio}
    \end{subfigure}
    \caption{Comparison of Chernoff information (CI) for core-periphery networks between ASE and DASE across varying overall network density ($\alpha_1$), sparse network density ($\alpha_2$) and varying group proportions ($\pi$).}
    \label{Fig:chernoff}
\end{figure}



The results of the experiments are shown in Figure \ref{Fig:chernoff}. When examining the effect of network density, we observe that the Chernoff information of both ASE and DASE increases as the network becomes denser. However, the increase for ASE is much smaller in magnitude than that for DASE, so the ASE curve appears nearly constant on the plotted scale (see Figures\ref{fig:chernoff_density_overall} and\ref{fig:chernoff_density}).

As observed in Figure \ref{fig:cheroff_ratio}, the Chernoff information of both ASE and DASE decreases as the core group becomes dominant. However, the Chernoff information for ASE appears nearly constant in the figure, since its magnitude is substantially smaller than that of DASE. In contrast, the Chernoff information for DASE remains consistently larger than that of ASE across all group proportions, even when the core group is dominant $(\pi_1 = 0.9)$. This indicates that the cluster distributions induced by the DASE embedding are more clearly separated than those obtained from ASE.


\subsection{Computational Cost Comparison}
In addition to clustering performance and embedding assessment with the Chernoff information, in this section we provide an analysis of the computational cost of the various embedding algorithms used in our simulation study. 

To this end, we use the same simulation settings as in the fixed expected network density ($\alpha = 0.05$) for a balanced two-cluster network ($\pi = (0.5, 0.5)^\top$) for directed graphs. We run each clustering method on graphs of increasing size, and record the runtime of each method, averaged over $n_{rep}=50$ simulated graphs of each particular network size. We then compute the average computation time relative to the average computation time for $N=500$ nodes.  This allows us to compare how the runtime of each method scales with the graph size, $N$, removing external factors such as efficiency of coding and implementation.

\begin{figure}[!h]
    \centering
        \includegraphics[width=0.4\linewidth]{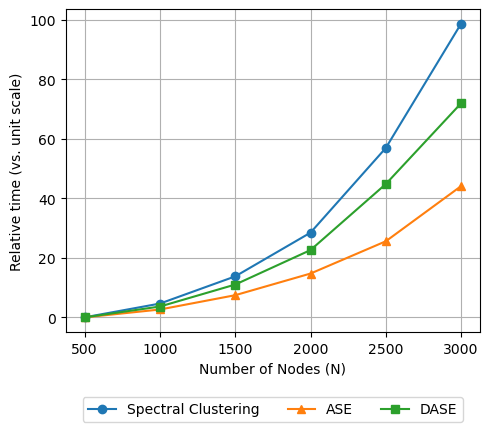}
    \caption{Comparison of computational complexity of clustering methods over $n_{rep}=50$ simulations, where the runtime for a network size of $N=500$ is set as the baseline, for $K=2$ with $\pi = (0.5, 0.5)^\top$. Here, we fix the expected network density ($\alpha = 0.05$) and vary the network size for directed graphs.}
    \label{fig:Fig_Comp_Relative}
\end{figure}


The results of the runtime analysis are shown in Figure \ref{fig:Fig_Comp_Relative}. As expected, the relative runtime of each method increases with the size of the simulated graphs under consideration. Perhaps unsurprisingly, \dasec\ increases faster than \asec, due to the additional computation required to form $\tilde{A} = AA$. Nevertheless, its scaling profile remains reasonably comparable to that of \asec. Given the improvement in clustering performance over competing methods across the range of examples studied in the simulations, we believe that this computational profile justifies the use of our proposed method.



For the undirected setting, the computational profile of the clustering methods is similar to that observed in the directed case shown in Figure \ref{fig:Fig_Comp_Relative}. However, the relative runtime is slightly lower than in the directed setting, as the embedding matrix $S$ in both \asec\ (Algorithm \ref{alg:ASE}) and \dasec\ (Algorithm \ref{alg:DASE}) has dimension $N \times d$ for undirected graphs, compared to $N \times 2d$ in the directed case.

\section{Application of \dasec\ to Real Data}
\label{sec:data}
In this section, we apply our graph clustering methodology to real-world datasets that are widely believed to exhibit core-periphery structural characteristics. We evaluate clustering accuracy using Normalized Mutual Information (NMI) when the ground truth is available, and use Chernoff information to identify the preferred embedding for each dataset; for the Chernoff information analysis, we convert the networks under consideration to be undirected in order to apply results from Section \ref{sec:chernoff}.  Note that for the analyses below, we also remove self-loops as per the assumption in Section \ref{sec:background}.

\subsection{Faculty Hiring Network}
\label{subsec:faculty}
The first dataset we consider is a faculty hiring network, described in \cite{clauset2015systematic}, containing hiring information from three different disciplines: Business, Computer Science, and History. For our analysis, we focus on the Computer Science dataset, which contains information on gender, faculty positions, and institutions from $23$ Canadian universities and $182$ American universities, collected from May 2011 to March 2012. In this network, each node represents an institution and each directed edge $(u, v)$ indicates that a faculty member at institution $v$ obtained their PhD degree from institution $u$; the network data comprises $N=205$ institutions with $|\mathcal{E}|=2,881$ unweighted directed hiring connections, resulting in a network density of $\alpha=0.0689$. For the Chernoff information analysis, this yields $|\mathcal{E}| = 2,693$ unweighted undirected edges and a corresponding network density of $\alpha = 0.1288$. The directed hiring network is shown in Figure \ref{fig:Fig8a}. \cite{clauset2015systematic} illustrate that the majority of individuals considered obtain their PhD degrees from the top $15\%$ of institutions (as measured by institutional prestige), representing a notional `core' of institutions, which exhibits higher interconnectivity within the network; hence \cite{clauset2015systematic} suggest that the faculty hiring network exhibits a core-periphery structure. The top 15\% prestige-ranked institutions in the dataset are shown in Figure \ref{fig:Fig8a} in yellow at the centre of the network, and reflect the findings of \cite{clauset2015systematic} that these `core' 
institutions produce more faculty members than other institutions.  The corresponding adjacency matrix is visualised in Figure \ref{fig:Fig8b}.



\begin{figure}[!h]
    \centering
    \begin{subfigure}{0.4\textwidth}
        \centering
        \includegraphics[scale=0.43]{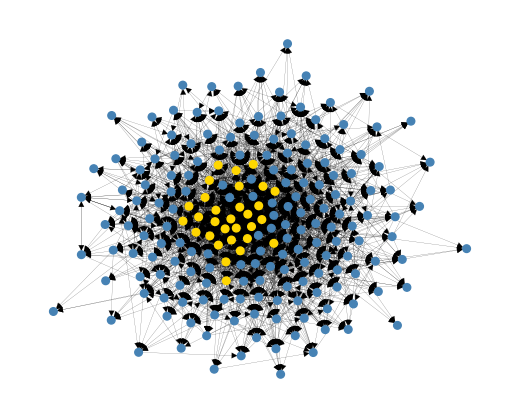}
        \caption{Network}
        \label{fig:Fig8a}
    \end{subfigure}
    \begin{subfigure}{0.4\textwidth}
        \centering
        \includegraphics[scale=0.43]{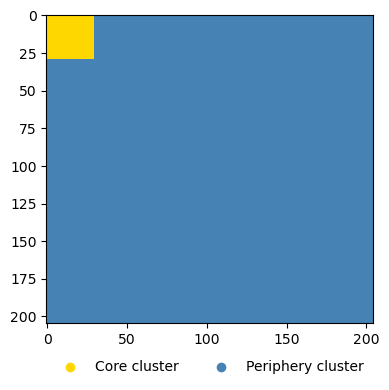}
        \caption{Heatmap}
        \label{fig:Fig8b}
    \end{subfigure}
    \caption{Faculty hiring network highlighting the top $15\%$ of universities: (a) the network visualisation with core universities shown in yellow, and (b) the corresponding adjacency matrix illustrating the core group.}
    \label{fig:Fig8}
\end{figure}

For the analysis of this dataset, since we expect one core group and one periphery group, the core representing the top $15\%$ institutions and the periphery the remaining ones, we assume that there are two clusters ($K=2$). We apply spectral clustering, \asec\ and \dasec\ (using $k$-means) to the hiring network. The corresponding clustering results are presented in Figure \ref{fig:Fig9}, where we show visualisations of the nodal cluster assignments and illustrate the core regions obtained from each clustering methods (in yellow). Inspecting Figures \ref{fig:Fig9a}-\ref{fig:Fig9c}, it is clear that spectral clustering fails to detect the core group. In contrast, \asec\ and \dasec\ have better clustering results to capture the core group, as indicated by the top-left corner of the clustering assignments.

\begin{figure*}[!h]
    \centering
    \begin{subfigure}{0.32\textwidth}
        \centering
        \includegraphics[scale=0.43]{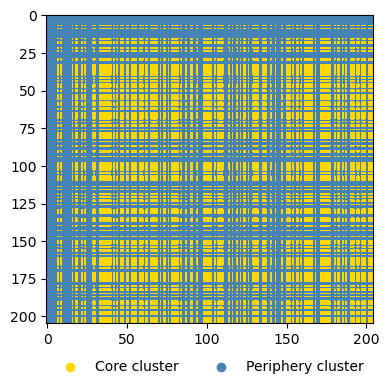}
        \caption{Spectral Clustering}
        \label{fig:Fig9a}
    \end{subfigure}
    \begin{subfigure}{0.32\textwidth}
        \centering
        \includegraphics[scale=0.43]{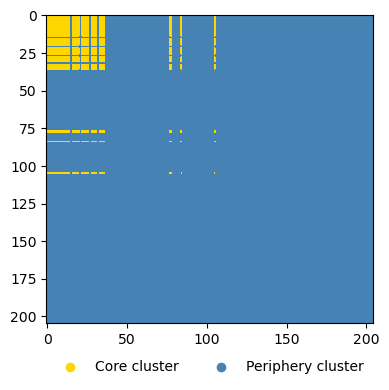}
        \caption{ASE}
        \label{fig:Fig9b}
    \end{subfigure}
    \begin{subfigure}{0.32\textwidth}
        \centering
        \includegraphics[scale=0.43]{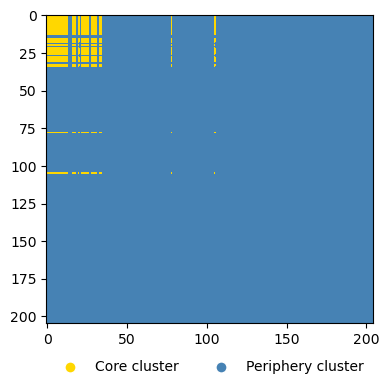}
        \caption{DASE}
        \label{fig:Fig9c}
    \end{subfigure}


    \caption{Visualisations of clustering results of universities in the faculty hiring network, based on clustering results obtained from spectral clustering, \asec\ and \dasec\ embeddings, with $K=2$ and $k$-means clustering. The the top $15\%$ university `core' is shown in yellow for each method.}
    \label{fig:Fig9}
\end{figure*}

Having compared the clustering results using the assignment visualisations, we now compare the clustering performance based on Normalized Mutual Information (NMI) on the directed hiring network and Chernoff information on the undirected hiring network, defined in \eqref{eq:Chernoff_Information}, using both $k$-means and GMM techniques; the results are presented in Table \ref{tab:faculty_entire}. Since our main interest lies in comparing performance between \asec\ and \dasec, we focus on these two methods.

For both NMI and Chernoff information, we repeat the clustering $30$ times with different random initialisations (i.e., different \texttt{random\_state} values in Python) for both $k$-means and GMM, while keeping the embedding fixed, and report the mean and standard deviation of the resulting NMI and Chernoff information values. In terms of NMI, \dasec\ performs better than \asec\ with $k$-means, whereas \asec\ performs better than \dasec\ with GMM. For the Chernoff information, DASE has higher values than ASE with both $k$-means and GMM. indicating that when combined with either clustering technique ($k$-means and GMM), DASE has higher separability in the spectral embedding space than ASE.

To further compare the improvement in NMI and Chernoff information, we consider the ratios between the values of the Chernoff information for DASE and ASE, defined as 
\begin{align*}
    \text{NMI}_{\text{ratio}} &= \frac{\text{NMI}_{\text{DASE}}}{\text{NMI}_{\text{ASE}}} \quad \text{and} \qquad
    CI_{\text{ratio}} = \frac{CI_{\text{DASE}}}{CI_{\text{ASE}}}
\end{align*}
A ratio greater than $1$ indicates that DASE achieves higher NMI or Chernoff information than ASE, suggesting improved clustering performance and better community separation in the embedding space.

The corresponding ratio table is in Table \ref{tab:faculty_ratio}. For NMI, as expected, using $k$-means yields greater than $1$, while using GMM yields smaller than $1$. On the other hand, for the Chernoff information, all ratio values, $CI_{\text{ratio}}$s, are greater than $1$. Moreover, the ratios are at least $29.987$, indicating that DASE exhibits much higher separability in the spectral embedding space than ASE. 

\begin{table*}[!h]
\centering
\resizebox{0.99\textwidth}{!}{%
\begin{tabular}{
   c
   c c c c
   c c c c
}
\toprule
 & \multicolumn{4}{c}{\textbf{NMI}} & \multicolumn{4}{c}{\textbf{Chernoff Information}} \\
\cmidrule(lr){2-5} \cmidrule(lr){6-9}
& \multicolumn{2}{c}{$k$-means} & \multicolumn{2}{c}{GMM} & \multicolumn{2}{c}{$k$-means} & \multicolumn{2}{c}{GMM} \\
\cmidrule(lr){2-3} \cmidrule(lr){4-5} \cmidrule(lr){6-7} \cmidrule(lr){8-9}
 & ASE & DASE & ASE & DASE & ASE & DASE & ASE & DASE \\
\midrule
& 0.5348 (0.00) & \textbf{0.5373} (0.02) & \textbf{0.3588} (0.00) & 0.3013 (0.00) & 0.0322 (0.01) & \textbf{1.0556} (0.00) & 0.0187 (0.00) & \textbf{0.5610} (0.00) \\

\bottomrule
\end{tabular}
}
\caption{Clustering performance (NMI and Chernoff information) of \asec\ and \dasec\ using $k$-means and GMM. For both NMI (directed network) and Chernoff information (undirected network), the reported mean and standard deviation are computed over $30$ clustering runs with different random initialisations. }
\label{tab:faculty_entire}
\end{table*}

\begin{table}[!h]
\centering
\begin{tabular}{@{\extracolsep{\fill}}ccccc@{}}
\toprule
 & \multicolumn{2}{c}{\textbf{NMI}} & \multicolumn{2}{c}{\textbf{Chernoff Information}} \\
\cmidrule(lr){2-3} \cmidrule(lr){4-5}
& \textbf{$k$-means} & \textbf{GMM} & \textbf{$k$-means} & \textbf{GMM} \\
\midrule
& 1.0047  & 0.8397 & 31.7726 & 29.9870 \\
\bottomrule
\end{tabular}
\caption{Ratios of DASE to ASE (DASE/ASE) for NMI (directed network) and Chernoff information (undirected network) using $k$-means and GMM.}
\label{tab:faculty_ratio}
\end{table}

\subsection{Air Traffic Network}
\label{subsec:air_traffic}

Our second real-world example is the crowdsourced air traffic data from the OpenSky Network during 2019 \parencite{strohmeier2021crowdsourced}. The dataset contains $N=15,439$ airports, including both commercial and private operations and $|\mathcal{E}| = 801,748$ flight connections, resulting in an edge density of $\alpha=0.0034$. Since flight connections have inherent origins and destinations, this network is directed. In this dataset, multiple connections exist between pairs of airports, making the network weighted. However, since the SBM model in \eqref{eq:sbm} assumes a binary adjacency matrix, we convert the weighted edges into binary connections, similar to \cite{gallagher2024spectral}, i.e., the edge $(u,v)$ indicates that at least one flight was made from origin $u$ to destination $v$.

\begin{figure}[!h]
    \centering
    \includegraphics[width=0.7\linewidth]{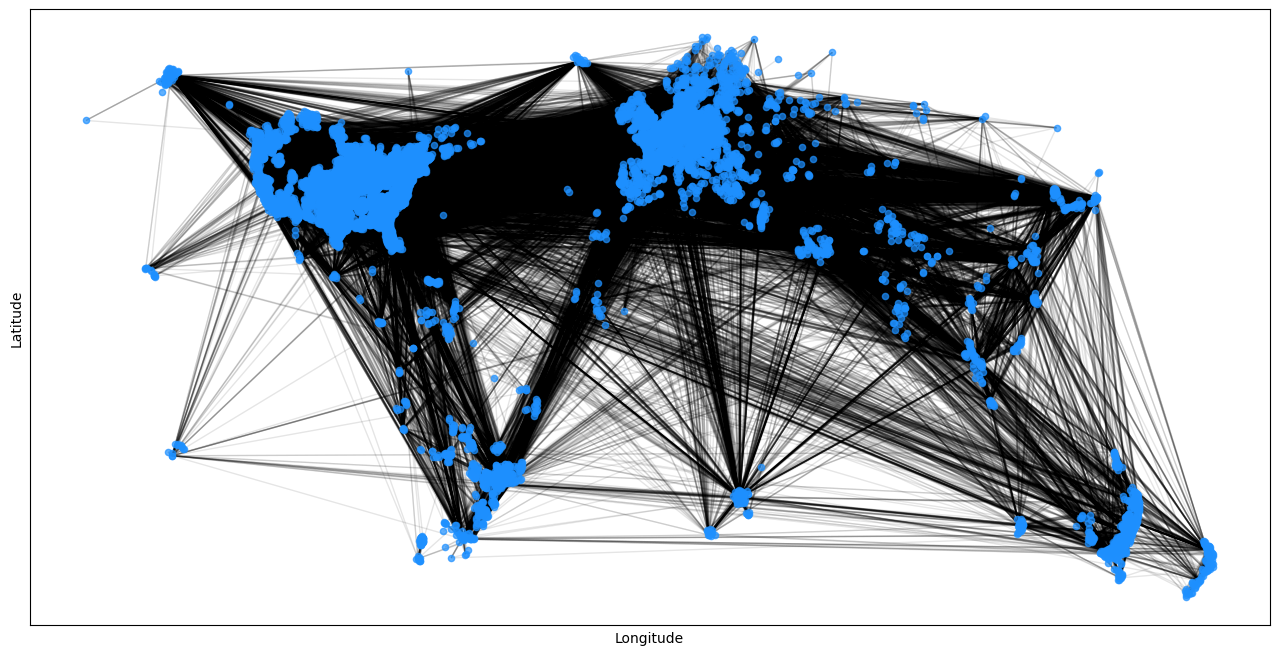}
    \caption{OpenSky air traffic network for 2019. Each node represents an airport, and each directed edges represents a flight connection.}
    \label{fig:Fig10}
\end{figure}

In our analysis, we perform the graph clustering procedures, considering two different values for the number of clusters. To identify one core group and one periphery group, we assume two clusters ($K=2$); as an alternate setting, we also consider six clusters ($K=6$), as suggested by \cite{gallagher2024spectral}. For the embedding dimension, we use $d=4$ as suggested by the profile likelihood method of \cite{zhu2006automatic}.

In contrast to the faculty hiring network in Section \ref{subsec:faculty}, there is no natural ground truth assessing clustering of the core and periphery nodes. Thus, it is not appropriate to define the ground truth using the highest-degree airports as core groups, since private airports and flights are also included, which distorts the network structure and prevents a meaningful labelling of the true core groups. Hence, we only examine the Chernoff information on the undirected network for different values of $K$; similar to the faculty hiring network, we hence convert the directed traffic network into an undirected network, yielding $|\mathcal{E}| = 613,202$ flight connections and a network density of $\alpha = 0.0051$.

A comparison of \asec\ and \dasec\ using $k$-means and GMM is summarised in Table \ref{tab:opensky_entire}. As we can see, for each cluster setting, DASE obtains higher Chernoff information values regardless of whether $k$-means or GMM is used. Similar to the faculty hiring network, we compare the ratio of the Chernoff information of ASE and DASE, shown in Table \ref{tab:opensky_ratio}. As we can see, all the ratio values are at least $94.8418$, indicating that DASE exhibits much higher separability in the spectral embedding space than ASE. 

\begin{table}[!h]
\centering
\resizebox{0.49\textwidth}{!}{%
\begin{tabular}{ccccc}
\toprule
 & \multicolumn{2}{c}{\textbf{$k$-means}} & \multicolumn{2}{c}{\textbf{GMM}} \\
\cmidrule(lr){2-3} \cmidrule(lr){4-5}
\textbf{\(K\)} & {ASE} & {DASE} & {ASE} & {DASE} \\
\midrule
2  & 0.0454 (0.00) & \textbf{16.546} (0.00) & 0.0059 (0.00) & \textbf{1.6531} (0.00) \\
6 & 0.0013 (0.00) & \textbf{0.5396} (0.21) & 0.0006 (0.00) & \textbf{0.0537} (0.01) \\
\bottomrule
\end{tabular}}
\caption{Comparison of the Chernoff information of ASE and DASE at $K=2$ and $K=6$ using $k$-means and GMM. For the Chernoff information on the undirected network, the reported mean and standard deviation are computed over $30$ clustering runs with different random initialisations.}
\label{tab:opensky_entire}
\end{table}

\begin{table}[!h]
\centering
\begin{tabular*}{0.4\textwidth}{@{\extracolsep{\fill}}ccc@{}}
\toprule
 & \multicolumn{2}{c}{\textbf{Chernoff Information}} \\
\cmidrule(lr){2-3}
\textbf{\(K\)} & \textbf{$k$-means} & \textbf{GMM} \\
\midrule
2  & 364.529  & 280.4302 \\
6  & 430.3368 & 94.8418   \\
\bottomrule
\end{tabular*}
\caption{Ratios of DASE to ASE (DASE/ASE) for the Chernoff information on the undirected network for $K=2$ and $K=6$, using $k$-means and GMM.}
\label{tab:opensky_ratio}
\end{table}

\section{Conclusion}
\label{sec:concs}

In this article we proposed \dasec, a novel graph clustering method for directed and undirected networks, which has particular suitability for core-periphery graphs. Our methodology leverages embedding the {\em doubled} adjacency matrix, which represents the possible ways to reach a node from another in two steps, thereby capturing more structural information about the graph than the adjacency matrix. We establish theoretical results showing that our clustering procedure achieves desirable consistency properties in terms of cluster recovery, as well as illustrating that the embedding afforded by Doubled Adjacency Spectral Embedding (DASE) for core-periphery graphs provides an improvement over the established Adjacency Spectral Embedding (ASE) technique when measured by the associated Chernoff information.

In addition to the theoretical comparison, through a range of numerical experiments, we evaluated the performance of our proposed method under different scenarios, including varying network sizes, densities and cluster sizes. Our experiments show that our proposed procedure obtains better accuracy and stability than other established methods for graph clustering, namely spectral clustering and ASE, particularly when the graph is sparse. Finally, we demonstrated \dasec\ on datasets from different scientific areas to illustrate its applicability in real-world settings. A natural extension would be to adapt the embedding to higher-order powers of the adjacency matrix, which we envisage would be beneficial in graph clustering tasks exhibiting different clique and homophily structures.  



\section{Acknowledgements}
\label{sec8}
SP is supported by a scholarship from the EPSRC Centre for Doctoral Training in Statistical Applied Mathematics at Bath (SAMBa), under the project EP/S022945/1. MAN gratefully acknowledges support from EPSRC grant EP/X002195/1.

\printbibliography

\clearpage

\appendix
\setcounter{figure}{0}
\renewcommand{\thefigure}{\Alph{section}.\arabic{figure}}
\setcounter{equation}{0}
\renewcommand{\theequation}{\Alph{section}.\arabic{equation}}

\section{Proof of Theoretical Results}
\label{appendix:proofs}

In this appendix, we provide details of the proofs of the results established in the main text. 



\paragraph{Additional notation.}
\label{sec:appendix_notations}


In what follows, recall from Section \ref{sec:notations} that $K$ denotes the number of clusters in a graph, and $d$ is the latent space dimension under consideration. 

Let the singular value decompositions of the matrices $Q$ and $A$ be as follows:
\begin{align*}
    Q = U \Sigma V^\top \quad \text{and} \quad A=U' \Sigma' V'^\top,
\end{align*}
where $U, U', V, V' \in \mathbb{R}^{N \times d}$ and $\Sigma, \Sigma' \in \mathbb{R}^{d \times d}$.
The concatenations of these matrices are
\[
W = [U|V] \in \mathbb{R}^{N \times 2d} \quad \text{and} \quad W' = [U'|V'] \in \mathbb{R}^{N \times 2d}.
\]
Similarly, let the singular value decompositions of the matrices $\tilde{Q}$ and $\tilde{A}$ be as follows:
\begin{align*}
    \tilde{Q} = QQ = \tilde{U} \tilde{\Sigma} \tilde{V}^\top \quad \text{and} \quad \tilde{A} = AA = \tilde{U}' \tilde{\Sigma}' \tilde{V}'^\top,
\end{align*}
where $\tilde{U}, \tilde{U}', \tilde{V}, \tilde{V}' \in \mathbb{R}^{N \times d}$ and $\tilde{\Sigma}, \tilde{\Sigma}' \in \mathbb{R}^{d \times d}$.
Therefore, the concatenations of these matrices are denoted
\[
\tilde{W} = [\tilde{U}|\tilde{V}] \in \mathbb{R}^{N \times 2d} \quad \text{and} \quad \tilde{W}' = [\tilde{U}'|\tilde{V}'] \in \mathbb{R}^{N \times 2d}.
\]

The results in this paper are asymptotic in the number of nodes $N$. However, we will sometimes drop this explicit dependence on $N$ of the matrices $A^{(N)}$, $Q^{(N)}$, $\tilde{A}^{(N)}$ and $\tilde{Q}^{(N)}$ for notational convenience in the proofs; note that the SVD quantities above should also be interpreted as being dependent on $N$ as in Section \ref{sec:theory}.

\subsection{Proof of Theorem \ref{thm_dir}}
\label{sec:proofthm1}

\begin{proof}
We first establish some intermediate results which will aid our consistency argument; in particular, we consider properties of $\tilde{A}^{(N)}\tilde{A}^{(N)}\!\!\phantom{.}^\top$ in order to infer properties of $\tilde{A}^{(N)}$.  


\begin{proposition}
\label{proposition1}
    Let $A^{(N)} \in \{0,1\}^{N \times N}$ be a sequence of adjacency matrices of directed graphs as defined in Section \ref{sec:notations} and let $\tilde{A}^{(N)} = A^{(N)}A^{(N)} \in [N]^{N \times N}$ denote the corresponding directed doubled adjacency matrices. 
Then the following holds with probability at least $1-4N^{-2}$:
    \begin{equation}
         \lVert \tilde{A}^{(N)}\tilde{A}^{(N)}\!\!\phantom{.}^\top - \tilde{Q}^{(N)}\tilde{Q}^{(N)}\!\!\phantom{.}^\top \rVert_F \leq N(N-1)^2 \sqrt{\log N} 
    \end{equation}
    and
    \begin{equation}
         \lVert \tilde{A}^{(N)}\!\!\phantom{.}^\top \tilde{A}^{(N)} - \tilde{Q}^{(N)}\!\!\phantom{.}^\top\tilde{Q}^{(N)} \rVert_F \leq N(N-1)^2 \sqrt{\log N}  .
    \end{equation}
\end{proposition}
\medskip

\begin{proof} 

    Similar to \cite{sussman2012consistent}, we would like to apply Hoeffding's inequality to obtain a deviation bound for the elements of $\tilde{A}\tilde{A}^\top$ from their expectation. 
    
            We first note that the entries of $\tilde{A}\tilde{A}^\top$ can be expressed as
     \begin{align}\label{eq:Awijuv}
        (\tilde{A}\tilde{A}^\top)_{uv} = \sum_{i=1}^N \sum_{j=1}^N \sum_{w=1}^N  A_{ui} A_{iw} A_{vj} A_{jw} =: \sum_{i,j,w=1}^N \mathcal{A}_{ijw}^{(u,v)},
    \end{align}
where $0\leq \mathcal{A}_{ijw}^{(u,v)} \leq 1$.\\

    To apply Hoeffding's inequality, we proceed by examining subsets of indices $(i,j,w)$ for which the random variables $\mathcal{A}_{ijw}^{(u,v)}$ are conditionally independent given $\tilde{Q}$ (equivalently given $Q$) for fixed $u,v \in [N]$, and we wish to bound the deviation 
    \begin{align}\label{eq:deviation}
        (\tilde{A}\tilde{A}^\top)_{uv} - (\tilde{Q}\tilde{Q}^\top)_{uv} = \sum_{w=1}^N \sum_{i=1}^N \sum_{j=1}^N A_{ui} A_{iw} A_{vj} A_{jw} - \sum_{w=1}^N \sum_{i=1}^N \sum_{j=1}^N Q_{ui} Q_{iw} Q_{vj} Q_{jw}.
    \end{align}

    We consider the two cases $u\neq v$ and  $u=v$ separately.\\

\underline{Case I: $u\neq v$.}\newline

    For fixed $(u,v)$, let $\mathcal{I} = \left\{ w,u \right\}$ and $\mathcal{J} = \left\{ w,v\right\}$. We consider the following three subsets of $(i,j) \in [N]^2$:
    \begin{align}\nonumber
        (i)& \quad \mathcal{S}_1 = \left\{ (i,j) \in [N]^2 : i \in \mathcal{I} \text{ or } j \in \mathcal{J} \right\}\\ \label{eq:S1S2S3}
        (ii)& \quad \mathcal{S}_2 = \left\{ (i,j) \in [N]^2: (i,j) \notin \mathcal{S}_1, i=j \right\}\\ 
        \nonumber
        (iii)& \quad \mathcal{S}_3 = \left\{ (i,j) \in [N]^2: i \notin \mathcal{S}_1, i\neq j \right\}, 
    \end{align}
\medskip
where $\mathcal{S}_1 \cup \mathcal{S}_2 \cup \mathcal{S}_3 = [N]^2$.

    \begin{itemize}\item[(i)]
    For $(i,j) \in \mathcal{S}_1$, the expression of $\mathcal{A}_{ijw}^{(u,v)}$ in \eqref{eq:Awijuv} involves diagonal entries of the matrix $A$. Since in our set up, we do not allow self-loops (see Section \ref{sec:background}), we have that $A_{ii} = 0$; the deviation in \eqref{eq:deviation} restricted to $(i,j)\in \mathcal{S}_1$ is then of the form
    \begin{align}\label{eq:devS1a}
        0 - \sum_{(i,j) \in \mathcal{S}_1} \sum_{w=1}^N Q_{ui} Q_{iw} Q_{vj} Q_{jw}.
    \end{align}

When $w\in\{u,v\}$, $|\mathcal{S}_1| = 3N-2$ and when $w\notin\{u,v\}$, $|\mathcal{S}_1| = 4N-4$.  Since $0 \leq Q_{ab} \leq 1$ for all $a,b \in [N]$, we have
    \begin{align}\label{eq:devS1bounda}
       0 \leq \sum_{(i,j) \in \mathcal{S}_1} \sum_{w=1}^N \mathcal{Q}_{ijw}^{(u,v)} = \sum_{(i,j) \in \mathcal{S}_1} \sum_{w\in \{u,v\}} \mathcal{Q}_{ijw}^{(u,v)}+\sum_{(i,j) \in \mathcal{S}_1} \sum_{w\notin \{u,v\} } \mathcal{Q}_{ijw}^{(u,v)}\leq  (2N-2)^2+2N,
    \end{align}
    where $\mathcal{Q}_{wij}^{(u,v)} = \mathbb{E}\left( \mathcal{A}_{wij}^{(u,v)} \right) = Q_{ui} Q_{iw} Q_{vj} Q_{jw}$. 
  \item[(ii)]  For $(i,j) \in \mathcal{S}_2$, since $\mathbb{E}(A_{iw}^2 ) = Q_{iw}$, the deviation \eqref{eq:deviation} becomes
    \begin{align}\label{eq:devS2a}
        \sum_{i \notin \{u,v,w \}} \sum_{w=1}^N \left( A_{ui} A_{iw}^2 A_{vi} - Q_{ui} Q_{iw} Q_{vi} \right).
    \end{align}

Note that when $w\in\{u,v\}$, $|\mathcal{S}_2| = |[N]\setminus \{u,v,w \} | = N-2$ and when $w\notin\{u,v\}$, $|\mathcal{S}_2| = |[N]\setminus \{u,v,w \} |= N-3$.  The $\{A_{ui} A_{iw}^2 A_{vi}\}$ in \eqref{eq:devS2a} are a set of $2\times (N-2) + (N-2)(N-3) = (N-1)(N-2)$ bounded and conditionally independent random variables given $Q$, and thus by Hoeffding's inequality 

\begin{equation}\label{eq:hoeffdingS2a}
    \mathbb{P} \left( \lvert (\tilde{A} \tilde{A}^\top)_{uv} - (\tilde{Q} \tilde{Q}^\top)_{uv} \rvert \geq  \sqrt{2(N-1)(N-2) \log N} \mid Q \right) \leq 2N^{-4}.
\end{equation}
\medskip

\item[(iii)]    When $(i,j) \in \mathcal{S}_3$, we have
$        \mathbb{E} \left( A_{ui} A_{iw} A_{vj} A_{jw} 
        \right) = Q_{ui} Q_{iw} Q_{vj} Q_{jw}$.
    Therefore, the deviation for $\mathcal{S}_3$ is of the form 
    \begin{align}\label{eq:devS3a}
        \sum_{(i,j) \in \mathcal{S}_3} \sum_{w=1}^N \left( A_{ui} A_{iw} A_{vj} A_{jw} - Q_{ui} Q_{iw} Q_{vj} Q_{jw} \right). 
    \end{align}
Similar to (i) and (ii) above, when $w\in\{u,v\}$, $|\mathcal{S}_3| = (N-2)^2$ and when $w\notin\{u,v\}$, $|\mathcal{S}_3| = (N-2)(N-3) + 1$ (since $\mathcal{S}_3 = [N]^2 \setminus \{\mathcal{S}_1 \cup \mathcal{S}_2\}$). 
Since all indices are different, the $\{A_{ui} A_{iw} A_{vj} A_{jw}\}$ in \eqref{eq:devS3a} are a set of $2\times (N-2)^2 + (N-2)[(N-2)(N-3)+1] = (N-2)[(N-2)(N-1)+1]$ bounded and conditionally independent random variables given $Q$, and thus by Hoeffding's inequality 

\begin{equation}\label{eq:hoeffdingS3a}
    \mathbb{P} \left( \lvert (\tilde{A} \tilde{A}^\top)_{uv} - (\tilde{Q} \tilde{Q}^\top)_{uv} \rvert \geq \sqrt{2(N-2)[(N-1)(N-2)+1] \log N } \mid Q \right) \leq 2N^{-4}.
\end{equation}
\end{itemize}\medskip

Since $(\tilde{A} \tilde{A}^\top)_{uv} - (\tilde{Q} \tilde{Q}^\top)_{uv}$ is equal to
    \begin{align*}
       \sum_{(i,j) \in \mathcal{S}_1} \sum_{w=1}^N  \left( \mathcal{A}_{wij}^{(u,v)} - \mathcal{Q}_{wij}^{(u,v)} \right) +   \sum_{(i,j) \in \mathcal{S}_2} \sum_{w=1}^N\left( \mathcal{A}_{wij}^{(u,v)} - \mathcal{Q}_{wij}^{(u,v)} \right) + \sum_{(i,j) \in \mathcal{S}_3} \sum_{w=1}^N\left( \mathcal{A}_{wij}^{(u,v)} - \mathcal{Q}_{wij}^{(u,v)} \right), 
    \end{align*}
taking the union bound over $\mathcal{S}_1$, $\mathcal{S}_2$ and $\mathcal{S}_3$, we combine \eqref{eq:devS1bounda}, \eqref{eq:hoeffdingS2a} and \eqref{eq:hoeffdingS3a} to give
\begin{equation}\label{eq:unoteqv}
    \mathbb{P} \left( \lvert (\tilde{A}\tilde{A}^\top)_{uv} - (\tilde{Q}\tilde{Q}^\top)_{uv} \rvert^2 \geq (N-1)^4\log N \mid Q  \right) \leq 4N^{-4} \qquad \text{for sufficiently large $N$}.
\end{equation}
\medskip

\underline{Case II: $u=v$.}\newline

When $u=v$, we take a similar approach to Case I above. In particular, we let $\mathcal{I} = \mathcal{J} = \left\{ w,u \right\}$ and consider the same three subsets of $(i,j) \in [N]^2$, $\mathcal{S}_1$, $\mathcal{S}_2$ and $\mathcal{S}_3$, as defined in equations \eqref{eq:S1S2S3}. We again consider the deviations \eqref{eq:deviation} for each of these subsets.
    \begin{itemize}\item[(i)]
    When $w=u$, $|\mathcal{S}_1| = 2N-2$ and when $w\neq u$, $|\mathcal{S}_1| = 4N-4$; hence for $(i,j) \in \mathcal{S}_1$, again considering the presence of self-loops, the deviation has the form
    \begin{align}\label{eq:devS1b}
        0 - \sum_{(i,j) \in \mathcal{S}_1} \sum_{w=1}^N Q_{ui} Q_{iw} Q_{uj} Q_{jw},
    \end{align}
with
    \begin{align}\label{eq:devS1boundb}
       0 \leq \sum_{(i,j) \in \mathcal{S}_1} \sum_{w=1}^N \mathcal{Q}_{ijw}^{(u,v)} \leq 1\times (2N-1) + (N-1) \times 4(N-1) = 4(N-1)^2 + 2N-1.
    \end{align}
  \item[(ii)]  For $(i,j) \in \mathcal{S}_2$, the deviation \eqref{eq:deviation} becomes
    \begin{align}\label{eq:devS2b}
        \sum_{i \notin \mathcal{I}} \sum_{w=1}^N \left( A_{ui}^2 A_{iw}^2 - Q_{ui} Q_{iw}  \right),
    \end{align}
where $A^2_{ui}$ and $A^2_{iw}$ are conditionally independent Bernoulli random variables given $Q$ with expectation $Q_{ui} Q_{iw}$. By similar counting arguments to Case I (considering when $w=u$ and $w\neq u$), the summands $\{A_{ui}^2 A_{iw}^2\}$ in \eqref{eq:devS2b} form a set of $1\times (N-1) + (N-1)(N-2) = (N-1)^2$ bounded and conditionally independent random variables given $Q$, and thus by Hoeffding's inequality 

\begin{equation}\label{eq:hoeffdingS2b}
    \mathbb{P} \left( \lvert (\tilde{A} \tilde{A}^\top)_{uu} - (\tilde{Q} \tilde{Q}^\top)_{uu} \rvert \geq  \sqrt{2(N-1)^2 \log} N \mid Q \right) \leq 2N^{-4}.
\end{equation}
\medskip

\item[(iii)]    Similar to Case I, when $(i,j) \in \mathcal{S}_3$ the deviation is of the form 
    \begin{align}\label{eq:devS3b}
        \sum_{(i,j) \in \mathcal{S}_3} \sum_{w=1}^N \left( A_{ui} A_{iw} A_{uj} A_{jw} - Q_{ui} Q_{iw} Q_{uj} Q_{jw} \right). 
    \end{align}
and the random variables in \eqref{eq:devS3b} form a set of $1\times (N-1)(N-2) + (N-1)[(N-2)(N-3)] = (N-1)(N-2)^2$ bounded and conditionally independent random variables given $Q$, and thus by Hoeffding's inequality 

\begin{equation}\label{eq:hoeffdingS3b}
    \mathbb{P} \left( \lvert (\tilde{A} \tilde{A}^\top)_{uu} - (\tilde{Q} \tilde{Q}^\top)_{uu} \rvert \geq \sqrt{ 2(N-1)(N-2)^2 \log N } \mid Q \right) \leq 2N^{-4}.
\end{equation}
\end{itemize}


Thus taking the union bound over $\mathcal{S}_1$, $\mathcal{S}_2$ and $\mathcal{S}_3$, combining \eqref{eq:devS1boundb}, \eqref{eq:hoeffdingS2b} and \eqref{eq:hoeffdingS3b} we have
\begin{equation}\label{eq:ueqv}
    \mathbb{P} \left( \lvert (\tilde{A}\tilde{A}^\top)_{uu} - (\tilde{Q}\tilde{Q}^\top)_{uu} \rvert^2 \geq (N-1)^4 \log N \mid Q  \right) \leq 4N^{-4} \qquad \text{for sufficiently large $N$}.
\end{equation}

Finally, taking the union bound over all $(u,v)\in [N]^2$ using \eqref{eq:unoteqv} and \eqref{eq:ueqv}, we obtain the deviation bound for $\tilde{A}\tilde{A}^\top - \tilde{Q}\tilde{Q}^\top$ in Frobenius norm as  
\begin{equation}\label{eq:finalbound}
    \mathbb{P} \left( \lVert \tilde{A}\tilde{A}^\top - \tilde{Q}\tilde{Q}^\top \rVert^2_F \geq N^2(N-1)^4 \log N \mid Q  \right) \leq 4N^{-2}.
\end{equation}
\medskip

Therefore, with probability at least $1-4N^{-2}$, we have
\begin{equation}
    \lVert \tilde{A}\tilde{A}^\top - \tilde{Q}\tilde{Q}^\top \rVert_F \leq N(N-1)^2 \sqrt{\log N}.
\end{equation}

The proof for $\lVert \tilde{A}^\top \tilde{A} - \tilde{Q}^\top \tilde{Q} \rVert_F$ follows analogously.
    
\end{proof}
\bigskip
\noindent
Proposition \ref{proposition1} establishes consistency of $\tilde{A}\tilde{A}^T$ to its expectation, which will facilitate establishing results on the singular values of $\tilde{A}$. Based on the Proposition \ref{proposition1}, we establish the misclustering rate of DASE when the graph is directed in Theorem \ref{thm_dir} using the following intermediate lemmas.

\begin{lemma}\label{lemma1} For $\tilde{Q}$ as defined in Section \ref{sec:dase},
\begin{enumerate}
    \item     It always holds that $\sigma_{d+1}(\tilde{Q}) = 0$ and $\sigma_1(\tilde{Q}) \leq N^2$.
    \item  It almost always holds that $(\min_i \pi_i) \cdot \tilde{b} N^2 \leq \sigma_d(\tilde{Q})$.

\end{enumerate}
\end{lemma}
\medskip
\begin{proof} \text{}\\
    \begin{enumerate}
        \item 
        Since $Q \in [0,1]^{N \times N}$, $\tilde{Q} = QQ \in [0,N]^{N \times N}$. The entries of the non-negative matrix $\tilde{Q}\tilde{Q}^T$ are bounded by $N^3$ since
        \[
        (\tilde{Q} \tilde{Q}^T)_{ij} = \sum_{k=1}^N \tilde{Q}_{ik} \tilde{Q}_{jk} \leq N^3,
        \]
        and hence the rowsums of $\tilde{Q} \tilde{Q}^T$ are bounded by $N^4$. Therefore, we have
            $\sigma_1^2(\tilde{Q}) = \lambda_1(\tilde{Q}\tilde{Q}^T) \leq N^4$ and thus $\sigma_1(\tilde{Q}) \leq N^2$.
        Since the rank of the matrix $Q$ is at most $d$, the rank of $\tilde{Q}$ is at most $d$ as well, i.e., $\sigma_{d+1}(\tilde{Q}) = 0$.
    
        \item 
       
        Since $Q: = ZBZ^\top$ (Section \ref{sec:sbm}), using the definitions of $D_n$ and $\tilde{B}$ from Section \ref{sec:notations}, we have
        \begin{align*}
            \tilde{Q} = QQ &= ZBZ^\top ZBZ^\top \\
            &= Z BD_n B Z^\top \qquad \textrm{since}\ Z^\top Z = D_n = \text{diag}(n_1, \dots, n_K)\\
            &= Z \tilde{B} Z^\top \qquad \qquad \textrm{since}\ \tilde{B} = B D_n B  \\
            &= Z \tilde{L} \tilde{\Lambda}\tilde{R}^\top Z^\top, 
        \end{align*}
using the notations for the singular value decomposition of $\tilde{B}$ from Section \ref{sec:notations}, $\tilde{B} = \tilde{L} \tilde{\Lambda} \tilde{R}^\top$, with $\tilde{L}, \tilde{R} \in \mathbb{R}^{K \times d}$.\\

        Let us define $\tilde{X} = Z \tilde{L} \tilde{\Lambda}^{1/2}$, $\tilde{Y} = Z \tilde{R} \tilde{\Lambda}^{1/2}$, so that 
        $\tilde{Q}  = \tilde{X}\tilde{Y}^\top$.
        The non-zero eigenvalues of $(\tilde{X}\tilde{Y}^\top)(\tilde{X}\tilde{Y}^\top)^\top$ are the same as the non-zero eigenvalues of $\tilde{Y}^\top \tilde{Y} \tilde{X}^\top \tilde{X}$. \\
        
        We consider $\tilde{Y}^\top \tilde{Y}$ and $\tilde{X}^\top \tilde{X}$ separately.\\
        
        The smallest non-zero eigenvalue of $\tilde{Y}^\top \tilde{Y}$ is obtained as follows. First, we note that
        \begin{align*}
            \tilde{Y}^\top \tilde{Y} &= \tilde{\Lambda}^{1/2} \tilde{R}^\top Z^\top Z \tilde{R} \tilde{\Lambda}^{1/2} \\
            &= \tilde{\Lambda}^{1/2} \tilde{R}^\top D_n \tilde{R} \tilde{\Lambda}^{1/2}\\
&=\tilde{\Lambda}^{1/2} S \tilde{\Lambda}^{1/2},
        \end{align*}
        by letting $S = \tilde{R}^\top D_n \tilde{R}$.\\ 

        Let $x\in \mathbb{R}^{d}$ be a generic vector.  Then we have
        \begin{align}\nonumber
        x^\top \tilde{R}^\top D_n \tilde{R} x &= (\tilde{R}x)^\top D_n \tilde{R}x\\
            &=\nonumber y^\top D_n y \\[3pt] \nonumber
            &= \sum_{k=1}^d n_k \cdot y_k^2 \\
            &\geq N \cdot  \min_{i \in \{1, \dots, d\} } \pi_i \cdot \lVert y \rVert^2 = N \tilde{\pi} \lVert y \rVert^2 \qquad \textrm{since by definition}\  \tilde{\pi} = \min_{i \in \{1, \dots, d\} } 
        \pi_i.\label{eq:uscRineq}
        \end{align}
If we further define $y = \tilde{R}x$, then $\lVert y \rVert^2 = \lVert \tilde{R}x \rVert^2 = \lVert x \rVert^2$.  Thus examining the Rayleigh quotient of $S$, $\frac{x^\top \tilde{R}^\top D_n \tilde{R} x}{\lVert x \rVert^2}$, we find from \eqref{eq:uscRineq} that 
 $\lambda_d(S)\geq N \tilde{\pi}$ (see e.g., \cite{Horn_Johnson_1985}).
 
Hence
        \begin{align*}
            \lambda_d(\tilde{Y}^\top \tilde{Y}) &= \lambda_d(\tilde{\Lambda}^{1/2} S \tilde{\Lambda}^{1/2}) \\
            &= \min_{\lVert z \rVert=1} z^\top \tilde{\Lambda}^{1/2} S \tilde{\Lambda}^{1/2} z \\
            &\geq \lambda_d(S) \cdot \min_{\lVert z \rVert=1} \lVert \tilde{\Lambda}^{1/2}z \rVert^2 \\
            &\geq N \tilde{\pi} \cdot \lambda_d(\tilde{\Lambda}) = N \tilde{\pi} \cdot \lambda_d(\tilde{B}).
        \end{align*}
 
        By similar arguments, we also obtain $\lambda_d(\tilde{X}^\top \tilde{X}) \geq N \tilde{\pi} \cdot \lambda_d(\tilde{B})$.\\
        
        Now by the eigenvalue inequality for matrix products (see e.g. \cite{zhang2006eigenvalue}),
        \begin{align*}
            \lambda_d(\tilde{Y}^\top \tilde{Y} \tilde{X}^\top \tilde{X}) &\geq \lambda_d(\tilde{Y}^\top \tilde{Y}) \lambda_d(\tilde{X}^\top \tilde{X}) \\
            &\geq \left( N \tilde{\pi} \cdot \lambda_d(\tilde{B}) \right)^2.
        \end{align*}
        Here, we can write that $\lambda_d(\tilde{B}) = \tilde{b} N$, where $0<\tilde{b} \in \mathbb{R}$.
        Therefore, we have
        \begin{align*}
            \sigma_d^2(\tilde{Q}) \geq \left( \tilde{\pi} \tilde{b} N^2 \right)^2,
        \end{align*}
        where $\tilde{\pi} = \min_{i \in \{1, \dots, d\} } 
        \pi_i$.
        \end{enumerate}
\end{proof}
\bigskip\noindent 
The following lemmas prove consistency (up to an orthogonal rotation) of the embedding of $\tilde{A}$, which follow similar arguments as Lemma 2 and Lemma 3 of \cite{sussman2012consistent} for ASE.

\begin{lemma}
\label{lemma2} Let $\tilde{W}$ and $\tilde{W}'$ be defined as above. Then
    it almost always holds that there exists an orthogonal matrix $\tilde{O} \in \mathbb{R}^{2d \times 2d}$ such that
    \begin{equation}
        \lVert \tilde{W}\tilde{O} - \tilde{W}'\rVert_F \leq \frac{2 N(N-1)^2 \sqrt{ \log N}}{\sigma_d^2(\tilde{Q})} \leq \frac{2}{\tilde{\pi}^2 \tilde{b}^2} \frac{(N-1)^2}{N^2}\frac{\sqrt{\log N}}{N}.
    \end{equation}
\end{lemma}
\medskip
\begin{proof}
    Let $\mathcal{S} = (\frac{1}{2} \sigma_d^2(\tilde{Q}), \infty)$. It almost always holds that exactly $d$ eigenvalues of $\tilde{A}\tilde{A}^\top$ and $\mathbb{E}(\tilde{A}\tilde{A}^\top)$ are in $\mathcal{S}$. Additionally, the minimum distance $\delta$ between $\tilde{U}$ in $\mathcal{S}$ and any eigenvalue of $\tilde{U}'$ not in $S$ is $\delta > \sigma_d^2(\tilde{Q})$. 
    
    By similar arguments to \cite{sussman2012consistent}, using the Frobenius norm bound found in Proposition \ref{proposition1}, the Davis-Kahan Theorem (see e.g. \cite{yu2015useful}) implies that  there exists an orthogonal matrix $\tilde{O}_1 \in \mathbb{R}^{d \times d}$ such that
    \begin{equation*}
        \lVert \tilde{U}\tilde{O}_1 - \tilde{U}'\rVert_F \leq \frac{\sqrt{2}}{\delta} \lVert \tilde{A}\tilde{A}^\top - \tilde{Q}\tilde{Q}^\top \rVert_F \leq \frac{\sqrt{2} N (N-1)^2 \sqrt{\log N}}{\sigma_d^2(\tilde{Q})}.
    \end{equation*}
    Similarly, for $\tilde{A}^\top \tilde{A}$, there exists there exists an orthogonal matrix $\tilde{O}_2 \in \mathbb{R}^{d \times d}$ such that
    \[
    \lVert \tilde{V} \tilde{O}_2 - \tilde{V}'\rVert_F \leq \frac{\sqrt{2} N (N-1)^2 \sqrt{\log N}}{\sigma_d^2(\tilde{Q})}.
    \]
\medskip
    Taking $\tilde{O} \in \mathbb{R}^{2d \times 2d}$ as the direct sum of $\tilde{O}_1$ and $\tilde{O}_2$, we then have that
    \begin{align}\nonumber
        \lVert \tilde{W}\tilde{O} - \tilde{W}' \rVert_F &= \sqrt{\lVert \tilde{U}\tilde{O}_1 - \tilde{U}'\rVert^2_F + \lVert \tilde{V} \tilde{O}_2 - \tilde{V}'\rVert^2_F} \\[3pt] \label{eq:directsum}
        &\leq \frac{2 N(N-1)^2 \sqrt{ \log N}}{\sigma_d^2(\tilde{Q})}.
    \end{align}
    Since $\sigma_d^2(\tilde{Q}) \geq \left( \tilde{\pi} \tilde{b} N^2 \right)^2$ from Lemma \ref{lemma1}, we thus have
    \begin{align*}
        \lVert \tilde{W}\tilde{O} - \tilde{W}' \rVert_F \leq \frac{2}{\tilde{\pi}^2 \tilde{b}^2} \frac{(N-1)^2}{N^2}\frac{\sqrt{\log N}}{N}.
    \end{align*}
\end{proof}
\bigskip\noindent

\begin{lemma}\label{lemma3} Let $\tilde{X}$ and $\tilde{Y}$ be defined as in the proof of Lemma \ref{lemma1}, and let $\hat{\beta}$, $\tilde{\beta}$ be defined as in Assumption \ref{assum:3} in Section \ref{sec:notations}.
    For all $u,v$ such that $\tilde{X}_u \ne \tilde{X}_v$, $\lVert \tilde{U}_u - \tilde{U}_v \rVert \geq \hat{\beta} \sqrt{ \tilde{b}\tilde{\pi} }$. Similarly, for all  $u,v$ such that $\tilde{Y}_u \ne \tilde{Y}_v$, $\lVert \tilde{V}_u - \tilde{V}_v \rVert \geq \hat{\beta} \sqrt{ \tilde{b}\tilde{\pi} }$. Therefore,
    \begin{equation}
        \lVert \tilde{W}_u - \tilde{W}_v \rVert \geq \tilde{\beta} \frac{\sqrt{\sigma_d(\tilde{B}) \cdot \tilde{\pi} N}}{\sigma_1(\tilde{Q})} \geq \hat{\beta} \sqrt{ \tilde{b} \tilde{\pi} },
    \end{equation}
    for all $u,v$ such that $\theta(u) \ne \theta(v)$.
\end{lemma}
\medskip
\begin{proof}
    Define $\tilde{E} = \tilde{\Lambda}^{1/2} \tilde{R}^\top \in \mathbb{R}^{d \times K}$, $\tilde{G} = \tilde{X} \tilde{E} \in \mathbb{R}^{N \times K}$, and $\tilde{G}' = \tilde{G} D_n^{1/2} \in \mathbb{R}^{N \times K}$. Since $Q = Z B Z^\top$ and $
        \tilde{Q} := \tilde{X} \tilde{Y}^\top$ (see proof of Lemma \ref{lemma1}), we have
    \begin{align*}
        \tilde{Q} \tilde{Q}^\top &= \left( \tilde{X} \tilde{Y}^\top \right) \left( \tilde{X} \tilde{Y}^\top \right)^T \\
        &= \tilde{X} \tilde{Y}^\top \tilde{Y} \tilde{X}^\top \\
        &= \tilde{X} \left( \tilde{\Lambda}^{1/2}  \tilde{R}^\top Z^\top  Z \tilde{R} \tilde{\Lambda}^{1/2} \right) \tilde{X}^\top \\
        &= \tilde{X} \left( \tilde{\Lambda}^{1/2}  \tilde{R}^\top \right ) D_n  \left( \tilde{R} \tilde{\Lambda}^{1/2} \right) \tilde{X}^\top \\
        &= \tilde{X} \tilde{E} D_n  \tilde{E}^T \tilde{X}^\top\\ 
        &= \tilde{G} D_n \tilde{G}^\top 
        = \tilde{G}' \tilde{G}'^\top. 
    \end{align*}
    Moreover, using the SVD notations introduced at the start of this appendix,
    \begin{align*}
        \tilde{Q}\tilde{Q}^\top &= \tilde{U} \tilde{\Sigma} \tilde{V}^\top \tilde{V} \tilde{\Sigma} \tilde{U}^\top\\ 
        &= \tilde{U} \tilde{\Sigma}^2 \tilde{U}^\top \\
        &= \tilde{U}' \tilde{U}'^\top. 
    \end{align*}
    Thus, we have
    \begin{align*}
        \tilde{U}' \tilde{U}'^\top = \tilde{Q} \tilde{Q}^\top = \tilde{G}' \tilde{G}'^\top.
    \end{align*}
    Let $e \in \mathbb{R}^N$ be the vector whose $u$-th entry is 1, whose $v$-th entry is $-1$, and whose remaining entries are all zero. Then, we have
    \begin{align*}
        \lVert \tilde{G}'_u - \tilde{G}'_v \rVert^2 &= e^\top \tilde{G}'\tilde{G}'^\top e \\
        &= e^\top \tilde{U}'\tilde{U}'^\top e \\
        &= \lVert \tilde{U}'_u - \tilde{U}'_v \rVert^2.
    \end{align*}
    Therefore, from the definition of $\tilde{\beta},$
    \begin{align*}
        \tilde{\beta} &< \lVert \tilde{X}_u - \tilde{X}_v \rVert \\
        &\leq \frac{1}{\sqrt{\sigma_d(\tilde{B})}} \lVert \tilde{G}_u - \tilde{G}_v \rVert \qquad \qquad \textrm{since}\  \tilde{G} = \tilde{X}\tilde{\Lambda}^{1/2} \tilde{R}^T \\
        &\leq \frac{1}{\sqrt{\sigma_d(\tilde{B}) \cdot \min_{i \in \{ 1, \dots, d \}} n_i}} \lVert \tilde{G}'_u - \tilde{G}'_v \rVert \qquad \textrm{since}\ \tilde{G}' = \tilde{G} D_n^{1/2} \\
        &= \frac{1}{\sqrt{\sigma_d(\tilde{B}) \cdot \tilde{\pi} N}} \lVert \tilde{U}'_u - \tilde{U}'_v \rVert\\ 
        &\leq \frac{\sigma_1(\tilde{Q})}{\sqrt{\sigma_d(\tilde{B}) \cdot \tilde{\pi} N}} \lVert \tilde{U}_u - \tilde{U}_v \rVert \qquad \qquad \textrm{since}\ \tilde{U}' = \tilde{U} \tilde{\Sigma}.
    \end{align*}
    Therefore, we finally have
    \begin{align*}
        \lVert \tilde{U}_u - \tilde{U}_v \rVert &\geq \tilde{\beta} \frac{\sqrt{\sigma_d(\tilde{B}) \cdot \tilde{\pi} N}}{\sigma_1(\tilde{Q})} \\
        &\geq \hat{\beta} N \frac{ \sqrt{\tilde{b} \tilde{\pi} N^2} }{N^2} = \hat{\beta} \sqrt{ \tilde{b} \tilde{\pi} },
    \end{align*}
    for all $u,v$ such that $\theta(u) \ne \theta(v)$, where in the last inequality we have used the bounds on $\sigma_d(\tilde{B})$ and $\sigma_1(\tilde{Q})$ from Lemma \ref{lemma1}.
    \\
    
    Similar to the argument based on $\tilde{Q}\tilde{Q}^\top$, we also obtain that for all $u,v$ such that $\tilde{Y}_u \ne \tilde{Y_v}$, $\lVert \tilde{V}_u - \tilde{V}_v \rVert \geq \hat{\beta} \sqrt{ \tilde{b}\tilde{\pi} }$, in terms of considering $\tilde{Q}^\top \tilde{Q}$.
    \\
    
    If nodes $u$ and $v$ belong to different groups, then they are already separated by either the left or right singular vectors $\tilde{U}$ or $\tilde{V}$. Hence, for $\lVert \tilde{W}_u - \tilde{W}_v \rVert$, if $\theta(u) \ne \theta(v)$, then it holds either $\tilde{U}_u \ne \tilde{U}_v$ or $\tilde{V}_u \ne \tilde{V}_v$. Therefore, we have
    \[
    \lVert \tilde{W}_u - \tilde{W}_v \rVert \geq \hat{\beta} \sqrt{ \tilde{b}\tilde{\pi}}.
    \]

\end{proof}
\bigskip\noindent

Finally, the proof of Theorem \ref{thm_dir} on the misclustering rate of \dasec\ is as follows.\\

    Let $\hat{\nu}$ and $\hat{\theta}$ satisfy the clustering criterion for $\tilde{W}$ (where $\tilde{W} = [\tilde{U}|\tilde{V}]$ takes the role of $\hat{X}$ in Equation (\ref{eq:clust})). Let $C' \in \mathbb{R}^{N \times 2d}$ have row $u$ given by $C'_u = \hat{\nu}_{\theta(u)}$. From Equation (\ref{eq:clust}), we can say that
    \begin{equation}
        \label{eq:C'}
            \lVert C' - \tilde{W}'\rVert_F \leq \lVert \tilde{W} \tilde{O} - \tilde{W}'\rVert_F.
    \end{equation}
    We then have
    \begin{align} \nonumber
        \lVert C' - \tilde{W} \tilde{O} \rVert_F 
        &= \lVert C' - \tilde{W}' + \tilde{W}' - \tilde{W} \tilde{O} \rVert_F \\ \nonumber
        &\leq \lVert C' - \tilde{W}' \rVert_F + \lVert \tilde{W}' - \tilde{W} \tilde{O} \rVert_F \\ \nonumber
        &\leq \lVert \tilde{W} \tilde{O} - \tilde{W}' \rVert_F + \lVert \tilde{W}' - \tilde{W} \tilde{O} \rVert_F 
            \quad \textrm{using}\ \eqref{eq:C'} \\ \nonumber
        &= 2 \times \lVert \tilde{W} \tilde{O} - \tilde{W}' \rVert_F \\ \label{eq:C'F}
        &\leq 2 \times \frac{2 N(N-1)^2 \sqrt{\log N}}{\sigma_d^2(\tilde{Q})},
    \end{align}

using the bound from Lemma \ref{lemma2}.\\

Let $\mathcal{B}_1, \mathcal{B}_2, \dots, \mathcal{B}_K$ be balls of radius
    \[
    r = \frac{\tilde{\beta}}{3} \cdot \frac{\sqrt{\sigma_d(\tilde{B}) \cdot \tilde{\pi} N}}{\sigma_1(\tilde{Q})}
    \]
    each centered around the $K$ distinct rows of $\tilde{W}$. Similar to \cite{sussman2012consistent}, by Lemma \ref{lemma3}, these balls are almost always disjoint.
    \\
    
The condition for the number of misclustered nodes     $
    \min_{\rho \in H_K} | \{ u \in V : \theta(u) \ne \rho(\hat{\theta}(u)) \} |
    $ is thus equivalent to the number of rows $u$ which satisfy 
    $$
    \lVert C'_u - \tilde{W}_u \tilde{O} \rVert > r.
    $$

Let $\mathcal{M} = \{ u: \lVert C'_u - \tilde{W}_u \tilde{O} \rVert > r\}$.\\
        
Since 
    $
    \lVert C' - \tilde{W} \tilde{O} \rVert_F
    = \sqrt{ \sum_u \lVert C'_u - \tilde{W}_u \tilde{O} \rVert^2 },
    $
    it follows from \eqref{eq:C'F} that
    \begin{align*}
        \sqrt{\sum_u \lVert C'_u - \tilde{W}_u \tilde{O} \rVert^2 }
        &= \lVert C' - \tilde{W}\tilde{O} \rVert_F
        \leq \frac{2^2 N(N-1)^2 \sqrt{\log N}}{\sigma_d^2(\tilde{Q})}. \\
        \Rightarrow\quad
        \sum_u \lVert C'_u - \tilde{W}_u \tilde{O} \rVert^2
        &\leq \frac{2^4 N^2(N-1)^4 \log N}{\sigma_d^4(\tilde{Q})}.
    \end{align*}

    Then we have
    \begin{equation*}
    \sum_{u \in M} \lVert C'_u - \tilde{W}_u \tilde{O} \rVert^2 \geq |\mathcal{M}| r^2.
    \end{equation*}
    Therefore,
    \begin{equation*}
        |\mathcal{M}| r^2 \leq \sum_{u \in \mathcal{M}} \lVert C'_u - \tilde{W}_u \tilde{O} \rVert^2 \leq \sum_u \lVert C'_u - \tilde{W}_u \tilde{O} \rVert^2 \leq \frac{2^4 N^2(N-1)^4 \log N}{\sigma_d^4(\tilde{Q})}.
    \end{equation*}
    Hence, by substituting $
    r = \frac{\tilde{\beta}}{3} \cdot \frac{\sqrt{\sigma_d(\tilde{B}) \tilde{\pi} N}}{\sigma_1(\tilde{Q})}$, the number of the misculstered nodes is bounded as follows:
    \begin{align*}
        \lvert \mathcal{M} \rvert &\leq \frac{2^4 \, 3^2 \, \sigma_1^2(\tilde{Q}) N^2 (N-1)^4 \log N}{ \tilde{\beta}^2 \sigma_d(\tilde{B}) (\tilde{\pi} N) \sigma_d^4(\tilde{Q}) } \\
        &\leq \frac{2^4 \, 3^2 \, N^4 N^2 (N-1)^4 \log N}{ (\hat{\beta} N)^2 (\tilde{b}N) (\tilde{\pi}N) (\tilde{\pi} \tilde{b}N^2)^4 } \qquad \textrm{since}\ \sigma_1(\tilde{Q}) \leq N^2, \tilde{\beta} = \hat{\beta}N, \sigma_d(\tilde{B}) = \tilde{b}N, \ \textrm{and}\ \sigma_d(\tilde{Q}) \geq (\tilde{\pi} \tilde{b} N^2)  \\
        &= \frac{2^4 \, 3^2}{\hat{\beta}^2 (\tilde{b} \tilde{\pi})^5~} \cdot \frac{(N-1)^4}{N^4} \cdot \frac{\log N}{N^2},
    \end{align*}
    where the constants $\hat{b}, \tilde{b}, \tilde{\pi} > 0$ are not dependent on the network size $N$.
    
\end{proof}

\subsection{Proof of Corollary \ref{corrl_undirected}}\label{sec:proofcorrl_undirected}
\begin{proof}
    Since $A$ is undirected, the singular value decomposition of $\tilde{Q}$ takes the form $\tilde{Q} = \tilde{U} \tilde{\Sigma} \tilde{U}^\top$, and similarly $\tilde{A} = \tilde{U}' \tilde{\Sigma}' \tilde{U}'^\top$. Here, $\tilde{W} = \tilde{U} \in \mathbb{R}^{N \times d}$. Therefore, by the same arguments as in the proof of Lemma \ref{lemma2}, there exists an orthogonal matrix $\tilde{O}_1 \in \mathbb{R}^{d \times d}$ such that
    \begin{align*}
        \lVert \tilde{U} \tilde{O}_1 - \tilde{U}' \rVert_F \leq \frac{\sqrt{2} N(N-1)^2 \sqrt{\log N}}{\sigma_d^2(\tilde{Q})}.
    \end{align*}
    We can circumvent the direct sum \eqref{eq:directsum} (since $\tilde{V}$ is unnecessary for the form of $\tilde{W}$) to obtain the bound 
            \begin{align*}
        \lVert \tilde{W}\tilde{O} - \tilde{W}' \rVert_F &= \sqrt{\lVert \tilde{U}\tilde{O}_1 - \tilde{U}'\rVert^2_F}\\[3pt] 
        &\leq \frac{\sqrt{2} N(N-1)^2 \sqrt{ \log N}}{\sigma_d^2(\tilde{Q})}\\
        &\leq \frac{\sqrt{2}}{\tilde{\pi}^2 \tilde{b}^2} \frac{(N-1)^2}{N^2}\frac{\sqrt{\log N}}{N},
    \end{align*}
    by incorporating the eigenvalue bound from Lemma \ref{lemma1}.  In other words, the bound established by Lemma \ref{lemma2} only changes by a factor of $\sqrt{2}$ when moving from the directed to the undirected network case.  The proof of Theorem \ref{thm_dir} follows directly with this change (using $\tilde{W} = \tilde{U} \in \mathbb{R}^{N \times d}$) and thus it almost always holds that
    \begin{align*}
        \min_{\rho \in H_K} | \{ u \in V : \theta(u) \ne \rho(\hat{\theta}(u)) \} | \leq \frac{2^3 \, 3^2}{\hat{\beta}^2 (\tilde{b} \tilde{\pi})^5~} \cdot \frac{(N-1)^4}{N^4} \cdot \frac{\log N}{N^2},
    \end{align*}
i.e. the bound changes by a factor of $2$ when moving from the directed to the undirected network case.
\end{proof}

\subsection{Proof of Theorem \ref{thm_dir_core}}
\label{sec:proofthmcore}



In this section, we focus on the consistency of DASE for the specific case of core-periphery structure. Recall from \eqref{eq:Q} that the matrix $Q$ can be written
\[
Q = 
\begin{bmatrix}
    p J_{n_1 \times n_1} & q J_{n_1 \times n_2} \\
    r J_{n_2 \times n_1} & s J_{n_2 \times n_2}
\end{bmatrix}.
\]

Recall that from Section \ref{sec:cpclust} for the core-periphery structure, we assume that
\[
0 \leq s < q,r < p \leq 1.
\]

\begin{proof}
As in the general graph setting, we begin by proving Proposition \ref{proposition_core} for core-periphery graphs, which establishes consistency of $\tilde{A}\tilde{A}^T$ for the specific case of (directed) core-periphery networks.

\begin{proposition}
\label{proposition_core} 
    Let $A^{(N)} \in \{0,1\}^{N \times N}$ be a sequence of adjacency matrices of a directed graph, and let $\tilde{A}^{(N)} = A^{(N)}A^{(N)} \in [N]^{N \times N}$ denote the corresponding sequence doubled adjacency matrices. Then, the following holds with probability at least $1-4N^{-2}$:
    \begin{equation}
         \lVert \tilde{A}^{(N)}\tilde{A}^{(N)}\!\!\phantom{.}^\top - \tilde{Q}^{(N)}\tilde{Q}^{(N)}\!\!\phantom{.}^\top \rVert_F \leq (p^4 - s^4) N(N-1)^2 \sqrt{\log N}
    \end{equation}
    and
    \begin{equation}
         \lVert \tilde{A}^{(N)}\!\!\phantom{.}^\top\tilde{A}^{(N)} - \tilde{Q}^{(N)}\!\!\phantom{.}^\top\tilde{Q}^{(N)} \rVert_F \leq (p^4 - s^4) N(N-1)^2 \sqrt{\log N}.
    \end{equation}
\end{proposition}
\medskip
\begin{proof}
Similar to the proof of Proposition \ref{proposition1}, we drop explicit dependence on the nodes $N$, and obtain bounds for $\lVert \tilde{A} \tilde{A}^\top - \tilde{Q} \tilde{Q}^\top \rVert_F$ for the two cases $u \ne v$ and $u=v$ separately.  The proof differs from the proof of Proposition \ref{proposition1} in that our bounds will be derived from the inequalities on the probabilities $p$, $q$, $r$ and $s$ for core-periphery networks.
\\
\\
\underline{Case I: $u\neq v$.}\newline
\\
We use the three subsets $\mathcal{S}_1, \mathcal{S}_2$ and $\mathcal{S}_3$, where $\mathcal{S}_1 \, \cup \, \mathcal{S}_2 \, \cup \, \mathcal{S}_3 = [N]^2$, as defined in the proof of Proposition \ref{proposition1}  (Case I).
\begin{itemize}\item[(i)]
    For $(i,j) \in \mathcal{S}_1$, we showed in \eqref{eq:devS1a} that the deviation in \eqref{eq:deviation} restricted to $(i,j)\in \mathcal{S}_1$ is 
    \begin{align}\label{eq:devS1a_core}
        0 - \sum_{(i,j) \in \mathcal{S}_1} \sum_{w=1}^N Q_{ui} Q_{iw} Q_{vj} Q_{jw},
    \end{align}
where $|\mathcal{S}_1| = 3N-2$ when $w\in\{u,v\}$ and $|\mathcal{S}_1| = 4N-4$ when $w\notin\{u,v\}$.  Since $s^4 \leq Q_{ui} Q_{iw} Q_{vj} Q_{jw} \leq p^4$, we have
    \begin{align}\label{eq:devS1bounda_core}
       s^4 \big( (2N-2)^2+2N \big) \leq \sum_{(i,j) \in \mathcal{S}_1} \sum_{w=1}^N \mathcal{Q}_{ijw}^{(u,v)} = \sum_{(i,j) \in \mathcal{S}_1} \sum_{w\in \{u,v\}} \mathcal{Q}_{ijw}^{(u,v)}+\sum_{(i,j) \in \mathcal{S}_1} \sum_{w\notin \{u,v\} } \mathcal{Q}_{ijw}^{(u,v)}\leq  p^4 \big( (2N-2)^2+2N \big),
    \end{align}
    where $\mathcal{Q}_{wij}^{(u,v)} = \mathbb{E}\left( \mathcal{A}_{wij}^{(u,v)} \right) = Q_{ui} Q_{iw} Q_{vj} Q_{jw}$. 
  \item[(ii)]  For $(i,j) \in \mathcal{S}_2$, recall from \eqref{eq:devS2a} that 
the deviation \eqref{eq:deviation} becomes
    \begin{align}\label{eq:devS2a_core}
        \sum_{i \notin \{u,v,w \}} \sum_{w=1}^N \left( A_{ui} A_{iw}^2 A_{vi} - Q_{ui} Q_{iw} Q_{vi} \right)=\sum_{i \notin \{u,v,w \}} \sum_{w=1}^N \xi'_{iw},
    \end{align}

defining $\xi'_{iw} := A_{ui} A^2_{iw} A_{vi} - Q_{ui} Q_{iw} Q_{vi}$, and where we note that $\mathbb{E}(\xi'_{iw}) = 0$ and $-p^3 \leq \xi'_{iw} \leq 1-s^3$.

The $\{\xi'_{iw}\}$ in \eqref{eq:devS2a_core} are a set of $2\times (N-2) + (N-2)(N-3) = (N-1)(N-2)$ bounded and conditionally independent random variables given $Q$, and thus by Hoeffding's inequality

\begin{equation}\label{eq:hoeffdingS2a_core}
    \mathbb{P} \left( \lvert (\tilde{A} \tilde{A}^\top)_{uv} - (\tilde{Q} \tilde{Q}^\top)_{uv} \rvert \geq  (1-s^3+p^3)\sqrt{2(N-1)(N-2) \log N} \mid Q \right) \leq 2N^{-4}.
\end{equation}
\medskip

\item[(iii)]    When $(i,j) \in \mathcal{S}_3$, the deviation for $\mathcal{S}_3$ is of the form 
    \begin{align}\label{eq:devS3a_core}
        \sum_{(i,j) \in \mathcal{S}_3} \sum_{w=1}^N \left( A_{ui} A_{iw} A_{vj} A_{jw} - Q_{ui} Q_{iw} Q_{vj} Q_{jw} \right)=\sum_{(i,j) \in \mathcal{S}_3} \sum_{w=1}^N \xi''_{ijw}, 
    \end{align}
where $\xi''_{ijw} = A_{ui} A_{iw} A_{vj} A_{jw} - Q_{ui} Q_{iw} Q_{vj} Q_{jw}$. Similar to above, these random variables are bounded as $\mathbb{E}(\xi''_{ijw})=0$ and $-p^4 \leq \xi''_{iw} \leq 1-s^4$. Since $\{\xi''_{ijw}\}$ form a set of $(N-2)[(N-2)(N-1)+1]$ conditionally independent random variables given $Q$, and thus by Hoeffding's inequality 

\begin{equation}\label{eq:hoeffdingS3a_core}
    \mathbb{P} \left(  \lvert (\tilde{A} \tilde{A}^\top)_{uv} - (\tilde{Q} \tilde{Q}^\top)_{uv} \rvert  \geq (1-s^4+p^4)\sqrt{2(N-2)[(N-1)(N-2)+1] \log N } \mid Q \right) \leq 2N^{-4}.
\end{equation}
\medskip
\end{itemize}

Since 
    \begin{align*}
        (\tilde{A} \tilde{A}^\top)_{uv} - (\tilde{Q} \tilde{Q}^\top)_{uv} &= \sum_{(i,j) \in \mathcal{S}_1} \sum_{w=1}^N  \left( \mathcal{A}_{wij}^{(u,v)} - \mathcal{Q}_{wij}^{(u,v)} \right) +   \sum_{(i,j) \in \mathcal{S}_2} \sum_{w=1}^N\left( \mathcal{A}_{wij}^{(u,v)} - \mathcal{Q}_{wij}^{(u,v)} \right) + \sum_{(i,j) \in \mathcal{S}_3} \sum_{w=1}^N\left( \mathcal{A}_{wij}^{(u,v)} - \mathcal{Q}_{wij}^{(u,v)} \right), 
    \end{align*}
taking the union bound over $\mathcal{S}_1$, $\mathcal{S}_2$ and $\mathcal{S}_3$, combining \eqref{eq:devS1bounda_core}, \eqref{eq:hoeffdingS2a_core} and \eqref{eq:hoeffdingS3a_core}, we have
\begin{equation}\label{eq:unoteqv_core}
    \mathbb{P} \left( \lvert (\tilde{A} \tilde{A}^\top)_{uv} - (\tilde{Q} \tilde{Q}^\top)_{uv} \rvert^2 \geq (p^4 - s^4)^2(N-1)^4\log N \mid Q  \right) \leq 4N^{-4} \qquad \text{for sufficiently large $N$}.
\end{equation}
\medskip

\underline{Case II: $u=v$.}\newline
\\
Similar to Case II in the proof of Proposition \ref{proposition1}, we also consider the deviation \ref{eq:deviation} for each of subsets $\mathcal{S}_1, \mathcal{S}_2$ and $\mathcal{S}_3$ (with $u=v$).

\begin{itemize}\item[(i)]
    When $w=u$, $|\mathcal{S}_1| = 2N-2$ and when $w\neq u$, $|\mathcal{S}_1| = 4N-4$; hence for $(i,j) \in \mathcal{S}_1$, again considering the presence of self-loops, the deviation has the form
 $$       0 - \sum_{(i,j) \in \mathcal{S}_1} \sum_{w=1}^N Q_{ui} Q_{iw} Q_{uj} Q_{jw},$$
with
    \begin{align}\label{eq:devS1boundb_core}
       s^4 N_{\mathcal{S}_1} \leq \sum_{(i,j) \in \mathcal{S}_1} \sum_{w=1}^N \mathcal{Q}_{ijw}^{(u,v)} \leq 1\times (2N-1) + (N-1) \times 4(N-1) = p^4 N_{\mathcal{S}_1},
    \end{align}
    where $N_{\mathcal{S}_1}: = 4(N-1)^2 + 2N-1$.
  \item[(ii)]  Recall from \eqref{eq:devS2b} that for $(i,j) \in \mathcal{S}_2$, the deviation \eqref{eq:deviation} becomes
 $$      \sum_{i \notin \mathcal{I}} \sum_{w=1}^N \left( A_{ui}^2 A_{iw}^2 - Q_{ui} Q_{iw}  \right),
 $$
where $A^2_{ui}$ and $A^2_{iw}$ are conditionally independent Bernoulli random variables given $Q$ with expectation $Q_{ui} Q_{iw}$.
\\
\\
Define $\xi'''_{iw} := A^2_{ui} A^2_{iw} - Q_{ui} Q_{iw}$, where $-p^2 \leq \xi'''_{iw} \leq 1-s^2$ and $\xi'''_{iw}$ is zero mean for all $i,w$. \\
\\
By similar counting arguments to Case I (considering when $w=u$ and $w\neq u$), the summands $\{\xi'''_{iw}\}$ in \eqref{eq:devS2b} form a set of $1\times (N-1) + (N-1)(N-2) = (N-1)^2$ bounded, zero-mean conditionally independent random variables given $Q$, and thus by Hoeffding's inequality 

\begin{equation}\label{eq:hoeffdingS2b_core}
    \mathbb{P} \left(  \lvert (\tilde{A} \tilde{A}^\top)_{uu} - (\tilde{Q} \tilde{Q}^\top)_{uu} \rvert \geq  (1-s^2+p^2)\sqrt{2(N-1)^2 \log} N \mid Q \right) \leq 2N^{-4}.
\end{equation}
\medskip

\item[(iii)]    Similar to Case I, when $(i,j) \in \mathcal{S}_3$ the deviation is of the form 
    \begin{align}\label{eq:devS3b_core}
        \sum_{(i,j) \in \mathcal{S}_3} \sum_{w=1}^N \left( A_{ui} A_{iw} A_{uj} A_{jw} - Q_{ui} Q_{iw} Q_{uj} Q_{jw} \right)= \sum_{(i,j) \in \mathcal{S}_3} \sum_{w=1}^N \xi''''_{ijw}, 
    \end{align}
where similar to (ii) above, $\xi''''_{ijw}$ is zero mean and $-p^4 \leq \xi''''_{ijw} \leq 1-s^4$.

The summands in \eqref{eq:devS3b_core} form a set of $1\times (N-1)(N-2) + (N-1)[(N-2)(N-3)] = (N-1)(N-2)^2$ conditionally independent, bounded random variables given $Q$, and thus by Hoeffding's inequality 

\begin{equation}\label{eq:hoeffdingS3b_core}
    \mathbb{P} \left( \lvert (\tilde{A} \tilde{A}^\top)_{uu} - (\tilde{Q} \tilde{Q}^\top)_{uu} \rvert \geq (1-s^4+p^4)\sqrt{ 2(N-1)(N-2)^2 \log N } \mid Q \right) \leq 2N^{-4}.
\end{equation}
\end{itemize}
\medskip

By taking the union bound over $\mathcal{S}_1$, $\mathcal{S}_2$ and $\mathcal{S}_3$, combining \eqref{eq:devS1boundb_core}, \eqref{eq:hoeffdingS2b_core} and \eqref{eq:hoeffdingS3b_core}, we have
\begin{equation}\label{eq:ueqv_core}
    \mathbb{P} \left( \lvert (\tilde{A}\tilde{A}^\top)_{uv} - (\tilde{Q}\tilde{Q}^\top)_{uv} \rvert^2 \geq (p^4 - s^4)^2(N-1)^4\log N \mid Q  \right) \leq 4N^{-4} \qquad \text{for sufficiently large $N$}.
\end{equation}

Finally, taking the union bound over all $(u,v) \in [N]^2$ using \eqref{eq:unoteqv_core} and \eqref{eq:ueqv_core}, we obtain the deviation bound for $\tilde{A} \tilde{A}^\top - \tilde{Q}\tilde{Q}^\top$ in the Frobenius norm as
\begin{equation} \label{eq:finalbound_core}
    \mathbb{P} \left( \lVert \tilde{A}\tilde{A}^\top - \tilde{Q}\tilde{Q}^\top \rVert^2_F \geq (p^4 - s^4)^2 N^2(N-1)^4 \log N \mid Q  \right) \leq 4N^{-2}.
\end{equation}
\medskip

Therefore, with probability at least $1-4N^{-2}$, we have
\begin{equation}
    \lVert \tilde{A}\tilde{A}^\top - \tilde{Q}\tilde{Q}^\top \rVert_F \leq (p^4 - s^4)N(N-1)^2 \sqrt{\log N}.
\end{equation}
The proof for $\lVert \tilde{A}^\top \tilde{A} - \tilde{Q}^\top \tilde{Q} \rVert_F$ follows analogously.

\end{proof}
\bigskip\noindent
Note that the difference between the bounds in Proposition \ref{proposition1} and Proposition \ref{proposition_core} is simply a factor of $(p^4 - s^4)^2$.  Thus the rest of the proof of Theorem \ref{thm_dir_core} follows directly using the exact same arguments as the proof of Theorem \ref{thm_dir}, in particular invoking the results of Lemma \ref{lemma2} and Lemma \ref{lemma3} with this adjustment in the bounds.

\end{proof}

\subsubsection{Proof of Corollary \ref{corrl_core_ase}}
\label{appendix:core_ase}

\begin{proof}
In this section, we focus on the proof of Corollary \ref{corrl_core_ase}, which establishes the misclustering rate for ASE for core-periphery networks. The proof follows the same structure as that of DASE and is based on the general network ASE analysis in \cite{sussman2012consistent}.

\begin{proposition}
\label{proposition_ase}
    Let $A \in \{0,1\}^{N \times N}$ be an adjacency matrix of a directed graph. Let $Q = \mathbb{E}(A)$ and the entries of $Q$ is in a range $[s, p]$, where $0 \leq s < p \leq 1$. Suppose that $A$ is conditionally independent given $Q$. Then, the following holds with probability at least $1-2N^{-2}$:
    \begin{equation}
         \lVert AA^\top  - QQ^\top \rVert_F \leq (1-s^2+p^2) \sqrt{3 N^3 \log N},
    \end{equation}
    and
    \begin{equation}
         \lVert A^\top A - Q^\top Q \rVert_F \leq (1-s^2+p^2) \sqrt{3 N^3 \log N}.
    \end{equation}

\end{proposition}
\medskip
\begin{proof}
    Similar to \cite{sussman2012consistent}, we wish to bound the deviation
    \begin{equation}
        \label{eq:deviation_ase}
        (AA^\top)_{uv} - (QQ^\top)_{uv} = \sum_{w = 1}^N A_{uw}A_{vw} = \sum_{w=1}^N Q_{uw} Q_{vw}.
    \end{equation}
    We consider the two cases $u \ne v$ and $u=v$ separately.
    \\
    \\
    \underline{Case I: $u\neq v$.}\newline \\
    Since $A_{uw}A_{vw}$ is conditionally independent when $w \notin \{u,v\}$, we have
    \begin{equation}
        (AA^\top)_{uv} - (QQ^\top )_{uv} = \sum_{w \notin \{u,v\}} (A_{uw}A_{vw} - Q_{uw}Q_{vw}) - (Q_{uu}Q_{vu} + Q_{uv}Q_{vv}). 
    \end{equation}
    
    Define the zero-mean random variable $\mathcal{E}_{w} := A_{uw}A_{vw} - Q_{uw}Q_{vw}$ with $-p^2 \leq \mathcal{E}_w \leq 1-s^2$.
    \\

    By Hoeffding's inequality, we have
    \begin{equation} \label{eq:hoeff_ase}
        \mathbb{P} \left( \lvert (AA^\top)_{uv} - (QQ^\top)_{uv} \rvert^2 \geq 2(1-s^2+p^2)^2(N-2)\log N + 2(1-s^2+p^2)N +4 \right) \leq 2N^{-4}.
    \end{equation}

    \underline{Case II: $u= v$.}\newline \\
    When $u=v$, we have
    \begin{equation} \label{eq:ueqv_ase}
        \left( (AA^\top)_{uu} - (QQ^\top)_{uu} \right)^2 \leq (1-s^2)^2N^2.
    \end{equation}

    Finally, taking the union bound over all $(u,v) \in [N]^2$ using \eqref{eq:hoeff_ase} and \eqref{eq:ueqv_ase}, we obtain the deviation bound for $AA^\top - QQ^\top $ in the Frobenius norm as
    \begin{equation} \label{eq:frob_ase}
        \mathbb{P} \left( \lVert AA^\top  - QQ^\top \rVert_F^2 \geq 3(1-s^2+p^2)^2 N^3 \log N \right) \leq 2N^{-2},
    \end{equation}
    since $2(1-s^2+p^2)N^3 + (1-s^2)^2N^3 + 4N^2 \leq (1-s^2+p^2)N^3 \log N$ for sufficiently large $N$.
    \\
    \\
    Therefore, with probability $1-2N^{-2}$, we have
    \begin{equation}
        \lVert AA^\top  - QQ^\top \rVert_F \leq (1-s^2+p^2) \sqrt{3 N^3 \log N}.
    \end{equation}
    The proof for $\lVert A^\top A  - Q^\top Q \rVert_F$ follows analogously.
    
\end{proof}
\bigskip\noindent

\begin{lemma}
\label{lemma9}
    It almost always holds that there exists an orthogonal matrix $O \in \mathbb{R}^{2d \times 2d}$ such that
    \begin{align}
        \lVert WO - W' \rVert_F &\leq \frac{2(1-s^2+p^2)\sqrt{3 \, N^3 \log N}}{\sigma_2^2(Q)} \\
        &\leq \frac{2\sqrt{3} \, (1-s^2+p^2)}{\tilde{\pi}^2 b^2} \cdot \sqrt{\frac{\log N}{N}},
    \end{align}
    where $\lambda_d(B) \geq b$ with some constant $b>0$.
\end{lemma}

\begin{proof}
    Similar to (the proof of) Lemma \ref{lemma2}, using the result from Proposition \ref{proposition_ase}, $\lVert UR_1 - U' \rVert$ is bounded as
    \begin{align*}
        \lVert U O_1 - U' \rVert_F \leq \frac{\sqrt{2}}{\sigma_2^2(Q)} \lVert AA^\top - QQ^\top \rVert_F \leq 
        \frac{(1-s^2+p^2)\sqrt{2 \cdot 3 \, N^3 \log N}}{\sigma_2^2(Q)},
    \end{align*}
    where an orthogonal matrix $O_1 \in \mathbb{R}^{d \times d}$. For the bound $\lVert VO_2 - V' \rVert_F$,
    \begin{align*}
        \lVert VO_2 - V' \rVert_F \leq \frac{(1-s^2+p^2)\sqrt{2 \cdot 3 \, N^3 \log N}}{\sigma_2^2(Q)},
    \end{align*}
    where an orthogonal matrix $O_2 \in \mathbb{R}^{d \times d}$. Therefore, by the direct sum, there exists an orthogonal matrix $O \in \mathbb{R}^{2d \times 2d}$, we have
    \begin{align*}
        \lVert WO - W' \rVert_F \leq \frac{2(1-s^2+p^2)\sqrt{3 \, N^3 \log N}}{\sigma_2^2(Q)}.
    \end{align*}
    \bigskip
    
    Since $\sigma_2^2(Q) \geq \left( \tilde{\pi} b N \right)^2$ from Lemma 1 in \cite{sussman2012consistent}, we thus have
    \begin{equation*}
        \lVert WO - W' \rVert_F \leq \frac{2\sqrt{3} \, (1-s^2+p^2)}{\tilde{\pi}^2 b^2} \cdot \sqrt{\frac{\log N}{N}},
    \end{equation*}
    where $\lambda_d(B) \geq b$ with some constant $b>0$.
\end{proof}
\bigskip

We now complete the proof of Corollary \ref{corrl_core_ase} as follows.\\

Analogously to Lemma \ref{lemma3}, by Lemma $3$ in \cite{sussman2012consistent}, we have that $\lVert W_u - W_v \rVert \geq \beta \sqrt{b\tilde{\pi}} N^{-1/2}$ for all $u,v$ such that $\theta(u) \ne \theta(v)$, which means that we can define balls of radius $r = \frac{\beta}{3} \sqrt{b \tilde{\pi}} N^{-1/2}$ centered on the distinct rows of $W$ and these balls are disjoint. 

Following similar arguments to the proof of Theorem \ref{thm_dir_core} (see also Theorem 1 in \cite{sussman2012consistent}), we have 
    \[
    \sum_u \lVert C_u' - W_u O \rVert^2 \leq \frac{2^4 \, 3 \, (1-s^2+p^2)^2 \, N^3 \log N}{\sigma_2^4(Q)}.
    \]

    Considering $\mathcal{M} = \{ u: \lVert C'_u - W_u O \rVert > r \}$ be a set of the misclustered nodes, we have
    \[
    \sum_{u \in \mathcal{M}} \lVert C'_u - W_u O \rVert^2 \geq \lvert \mathcal{M} \rvert r^2.
    \]
    Therefore,
    \[
    \lvert \mathcal{M} \rvert r^2 \leq \sum_u \lVert C_u' - W_u O \rVert^2 \leq \frac{2^4 \, 3 \, (1-s^2+p^2)^2 \, N^3 \log N}{\sigma_2^4(Q)}
    \]
    By substituting
    \[
    r = \frac{\beta}{3} \sqrt{b \tilde{\pi}} N^{-1/2},
    \]
    the number of misclustered nodes is bounded as follows:
    \begin{align*}
        \lvert \mathcal{M} \rvert &\leq \frac{2^4 \, 3^3 \, (1-s^2+p^2)^2 \, N^3 \log N}{\beta^2 (b\tilde{\pi}) N^{-1} \sigma_2^4(Q)} \\
        &\leq \frac{2^4 \, 3^3 \, (1-s^2+p^2)^2 \, N^4 \log N}{\beta^2 (b\tilde{\pi}) (b\tilde{\pi}N)^4} \qquad (\because \sigma_d(Q) \geq b \tilde{\pi} N) \\
        &= \frac{2^4 \, 3^3 \, (1-s^2+p^2)^2}{\beta^2 (b\tilde{\pi})^5} \, \log N,
    \end{align*}
    where constants $b, \beta, \tilde{\pi}> 0$ are not dependent on the network size $N$.
\end{proof}

\section{Additional Simulation Results}
\label{appendix:sims}
In this section we provide details of additional simulation results using the Gaussian Mixture Model (GMM) in place of $k$-means in Step 6 (respectively Step 7) of \asec\ in Algorithm \ref{alg:ASE} (respectively \dasec\ in Algorithm \ref{alg:DASE}).  More specifically, Section \ref{appendix:directed} describes findings for directed graphs, whereas Section \ref{appendix:undirected} considers the undirected graph case. For both settings, similar to Section \ref{sec:sims}, the performance of the clustering methods are represented in figures as follows: spectral clustering (circle, blue); \asec\ (triangle, orange) and \dasec\ (square, green).

Similar to the $k$-means clustering simulations in Section \ref{sec:sims}, we consider three types of simulations: (1) fixing the network size and varying the network density with balanced networks, (2) fixing the network density and varying the network size with balanced networks, and (3) fixing the network size and density but varying the size of core groups ($\pi_1$).

\subsection{Directed Graphs}
\label{appendix:directed}
Figure \ref{fig:appendix_dir} shows the clustering performance measured by the average Normalised Mutual Information (NMI), along with the corresponding standard deviations (shaded areas) over $n_{rep}=50$ simulated graphs. When the network size is fixed and the network density varies, both \asec\ and \dasec\ recover the true labels even when the network density is sparse, as shown in Figure \ref{fig:appendix_Fig1a}. However, the standard deviation of \asec\ is larger than that of \dasec, indicating that \dasec\ is more stable in sparse settings.

For fixed network density and varying network sizes, the simulation results are shown in Figures \ref{fig:appendix_Fig1b}. As seen in Figure \ref{fig:appendix_Fig1b}, both \asec\ and \dasec\ outperform spectral clustering, as expected. Once again, the standard deviation (shaded area) of \asec\ is again larger than that of \dasec, confirming that \dasec\ provides a smaller clustering uncertainty than \asec.

\begin{figure*}[!h]
    \centering
    \begin{subfigure}{0.4\textwidth}
        \centering
        \includegraphics[width=\linewidth]{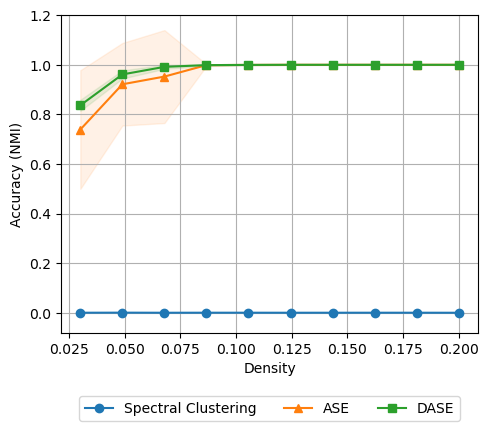}
        \caption{NMI with fixed network size}
        \label{fig:appendix_Fig1a}

    \end{subfigure}
    \begin{subfigure}{0.4\textwidth}
        \centering
        \includegraphics[width=\linewidth]{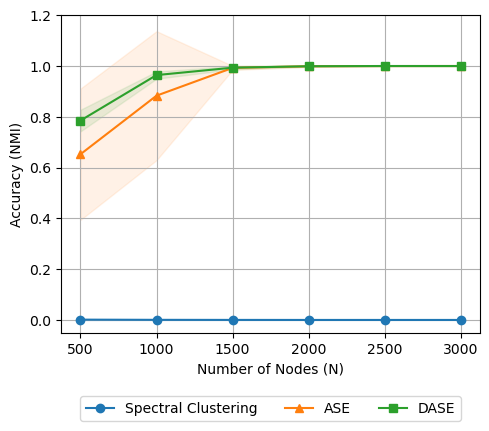}
        \caption{NMI with fixed network density}
        \label{fig:appendix_Fig1b}
    \end{subfigure}


    \caption{Figures illustrating the comparison of clustering performance on directed graphs using GMM in terms of mean accuracy (NMI), corresponding standard deviation over $n_{rep}=50$ iteration when $K=2$ with $\pi = (0.5, 0.5)^\top$: (a) NMI with fixed network size ($N=1,000$) and varying network density; (b) NMI with fixed expected network density ($\alpha = 0.05$) and varying network sizes. In (a) and (b), shaded areas represent the standard deviations of the NMI values over the $n_{rep}$ simulated graphs.}
    \label{fig:appendix_dir}
\end{figure*}

We also consider unbalanced two-cluster networks, as shown in Figure \ref{fig:appendix_Fig2}. Across different proportions of the core group, \asec\ almost recovers the true labels, but its performance fluctuates substantially, and its standard deviation is larger than that of \dasec. In contrast, DASE not only achieves better clustering accuracy than \asec\ but also remains more stable across different cluster size ratios. Across different core cluster proportions, \dasec\ exhibits its largest standard deviations and reduced clustering performance when $\pi_1 = 0.9$, where the network structure resembles a single dense group. However, overall these results confirm that \dasec\ remains stable even in unbalanced settings.

\begin{figure}[!h]
    \centering
    \includegraphics[width=0.4\linewidth]{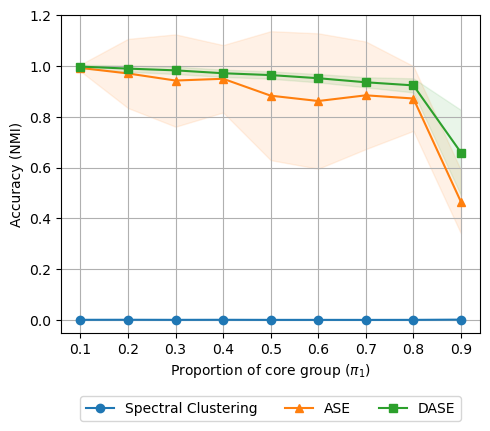}
    \caption{Figure showing the comparison of clustering performance on directed graphs using GMM in terms of mean NMI (line) and the corresponding standard deviation (shaded area) over $n_{rep}=50$ simlulated graphs when $K=2$. In the simulation, the network size is fixed at $(N=1,000)$, and the block probability matrix $B$ is fixed, while varying the core group ratio ($\pi_1$) from $0.1$ to $0.9$.}
    \label{fig:appendix_Fig2}
\end{figure}

\subsection{Undirected Graphs}
\label{appendix:undirected}
We next consider undirected, balanced two-cluster networks to evaluate the clustering performance of spectral clustering, \asec\ and \dasec\ for the GMM clustering technique.  The results are shown in Figure \ref{fig:appendix_Fig3}. Here, we focus on the clustering accuracy under fixed network size and fixed network density, presented in Figures \ref{fig:appendix_Fig3a} and \ref{fig:appendix_Fig3b}, respectively.

\begin{figure}[!h]
    \centering
    \begin{subfigure}{0.4\textwidth}
        \centering
        \includegraphics[width=\linewidth]{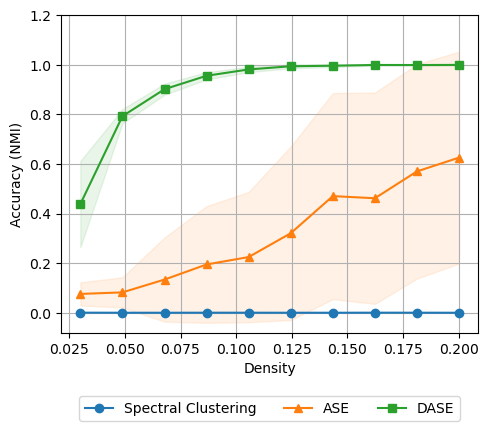}
        \caption{NMI with fixed network size}
        \label{fig:appendix_Fig3a}

    \end{subfigure}
    \begin{subfigure}{0.4\textwidth}
        \centering
        \includegraphics[width=\linewidth]{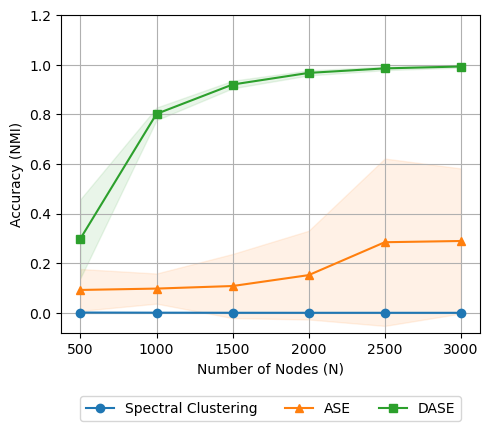}
        \caption{NMI with fixed network density}
        \label{fig:appendix_Fig3b}
    \end{subfigure}   

    \caption{Figures illustrating the comparison of clustering performance on undirected graphs using GMM in terms of mean accuracy (NMI), corresponding standard deviation, and mean computational cost over $50$ iteration when $K=2$ with $\pi = (0.5, 0.5)^\top$: (a) NMI with fixed network size ($N=1,000$) and varying network density; (b) NMI with fixed expected network density ($\alpha = 0.05$) and varying network sizes. In (a) and (b), shaded areas represent the standard deviations of the NMI values over the $n_{rep}$ simulated graphs.}
    \label{fig:appendix_Fig3}
\end{figure}

For fixed network size and varying network density in Figure \ref{fig:appendix_Fig3a}, \asec\ performs better than spectral clustering but worse than \dasec, even when the network is not particularly sparse. In addition, the standard deviation of the clustering performance for \asec\ is larger than that of \dasec. For fixed network density and varying network size, shown in Figure \ref{fig:appendix_Fig3b}, the results are consistent with those under fixed network size: \dasec\ successfully recovers the true label when $N>2,5000$, whereas the accuracy of \asec\ falls below $0.4$ when $N=3,000$.

We also examine unbalanced, undirected two-cluster networks, as shown in Figure \ref{fig:appendix_Fig4}. As expected, \dasec\ consistently outperforms \asec\ and spectral clustering across different core-group ratios. In contrast, clustering performance of \asec\ drops when the two groups are closer in size and increases again as the core group becomes dominant. \dasec, on the other hand, performs well except $\pi_1 \geq 0.7$, where the core group is dominant. Although its accuracy decreases in these cases, it still performs better than \asec\ across all settings.\\


\begin{figure}[ht]
    \centering
    \includegraphics[width=0.4\linewidth]{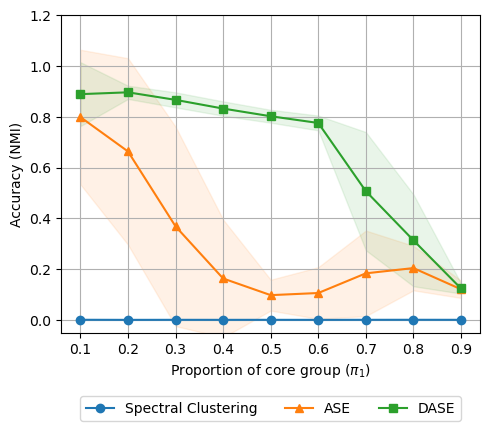}
    \caption{Figure presenting the comparison of clustering performance on undirected graphs using GMM in terms of mean NMI (line) and the corresponding standard deviation (shaded area) over $50$ iterations when $K=2$. In the simulation, the network size is fixed at $(N=1,000)$, and the block probability matrix $B$ is fixed, while varying the core group ratio ($\pi_1$) from $0.1$ to $0.9$.}
    \label{fig:appendix_Fig4}
\end{figure}

Similar to Section \ref{sec:sims}, we note here that the clustering methods achieve better clustering performance for directed graphs than for undirected graphs, due to the fact that converting a directed network into an undirected one leads to  loss of edge information, and thus this drop in performance is to be expected.

\end{document}